%% file: decattack.tex
\newif\ifarxiv
\newif\ifdouble
\newif\ifsingle
\newif\iffull

\arxivtrue        
\fulltrue         

\ifarxiv
\documentclass[11pt,a4paper]{article}
\fi
\ifdouble
\documentclass[10pt,twocolumn,twoside]{IEEEtran}
\fi
\ifsingle
\documentclass[11pt,draftcls,onecolumn]{IEEEtran}
\fi

\usepackage{setspace}
\usepackage{amsmath}
\usepackage{amssymb}
\usepackage{amsthm}
\usepackage{amsfonts}
\usepackage{cite}
\usepackage{tikz}
\usepackage{pgfplots}
\usepackage{subfigure}
\usepackage{wrapfig}
\usepackage{varioref}
\usepackage{hyperref}

\newtheorem{thm}{Theorem}
\newtheorem{lemma}{Lemma}

\newtheorem{cor}{Corollary}
\newtheorem{defn}{Definition}
\newtheorem{alg}{Algorithm}

\newcommand{\AlgInput }{\text{ }\newline\newline\textit{\em\bf Input}~~~~}
\newcommand{\AlgOutput}{\text{ }\newline\newline\textit{\em\bf Output}~~~} 
\newcommand{\AlgFailure}{\textsc{Failure}}

\newcommand{\FiniteField}{\text{${\bf F}$}}     
\newcommand{\FieldSize}{\text{${n}$}}     
\newcommand{\CharPoly}{\text{$\chi$}}           
\newcommand{\VeatureCoeff}{\text{$\xi$}}          
\newcommand{\SecretSize}{\text{$k$}}
\newcommand{\GenuinePairs}{\text{${\bf G}$}}
\newcommand{\ChaffPairs}{\text{${\bf C}$}}
\newcommand{\VaultPairs}{\text{$\bf V$}}
\newcommand{\UnlockingPairs}{\text{${\bf U}$}}
\newcommand{\GenuineSize}{\text{$t$}}
\newcommand{\VaultSize}{\text{$r$}}
\newcommand{\absc}{\text{$x$}}
\newcommand{\ord}{\text{$y$}}
\newcommand{\VeatureSet}{\text{$\bf A$}}         
\newcommand{\WeatureSet}{\text{$\bf B$}}         
\newcommand{\ZeatureSet}{\text{$\bf D$}}         

\newcommand{\VaultCoeff}{\text{$v$}} 	            
\newcommand{\WaultCoeff}{\text{$w$}} 	            
\newcommand{\ZaultCoeff}{\text{$z$}}                

\newcommand{\VaultPoly}{\text{$V$}}
\newcommand{\WaultPoly}{\text{$W$}}
\newcommand{\ZaultPoly}{\text{$Z$}}
\newcommand{\VeatureElement}{\text{$a$}}         
\newcommand{\WeatureElement}{\text{$b$}}         
\newcommand{\ZeatureElement}{\text{$c$}}         
\newcommand{\SYM}{\text{$X$}}
\newcommand{\NumOverlap}{\text{$\omega$}}
\newcommand{\VeatureSize}{\text{$t$}}
\newcommand{\WeatureSize}{\text{$s$}}
\newcommand{\RemPoly}{\text{$R$}}
\newcommand{\VaultSecret}{\text{$f$}}
\newcommand{\WaultSecret}{\text{$g$}}
\newcommand{\VErrorPoly}{\text{$P$}}
\newcommand{\WErrorPoly}{\text{$Q$}}
\newcommand{\VErrorSize}{\text{$\delta$}}                
\newcommand{\WErrorSize}{\text{$\epsilon$}}              
\newcommand{\EEAConstant}{\text{$\alpha$}}               
\newcommand{\EEAInverse}{\text{$\beta$}}                 
\newcommand{\CandidateIndicesList}{\text{$\mathcal{L}$}}
\newcommand{\RelationErrorPoly}{\text{$\mathcal{E}$}}

\newcommand{\FiniteFieldExt}{\text{${\bf K}$}}
\newcommand{\AdditionalVeatureSet}{\text{$\VeatureSet_{\operatorname{bl}}$}}
\newcommand{\AdditionalVeatureSize}{\text{$\VeatureSize_{\operatorname{bl}}$}}
\newcommand{\AdditionalWeatureSet}{\text{$\WeatureSet_{\operatorname{bl}}$}}
\newcommand{\AdditionalWeatureSize}{\text{$\WeatureSize_{\operatorname{bl}}$}}

\newcommand{\AverageMinEntropy}{\text{$\mathbf{\tilde{H}_{\infty}}$}}
\newcommand{\MinEntropy}{\text{$\mathbf{H_{\infty}}$}}
\newcommand{\Expectation}{\text{$\mathbb E$}}
\newcommand{\diffVW}{\text{$d$}}
\newcommand{\VaultSurr}{\text{$\mathcal U$}}              
\newcommand{\RandVarX}{\text{$X$}}
\newcommand{\RandVarY}{\text{$Y$}}
\newcommand{\varx}{\text{$x$}}
\newcommand{\vary}{\text{$y$}}
\newcommand{\Oh}{\text{$\mathcal{O}$}}
\newcommand{\SoftOh}{\text{$\mathcal{O}^\sim$}}
\newcommand{\rem}{\text{$~\operatorname{rem}~$}}

\newcommand{\VaultSpurChar}{\text{$V_0$}}
\newcommand{\WaultSpurChar}{\text{$W_0$}}
\newcommand{\SetDiff}{\mathrm{dist}}
\newcommand{\prob}[1]{\text{$p_\mathrm{\scriptscriptstyle\textsc{#1}}$}}
\newcommand{\AlgA}{\text{$\mathcal{A}$}}
\newcommand{\AlgB}{\text{$\mathcal{B}$}}
\newcommand{\ThreshDist}{\text{$h$}}
\newcommand{\LastGuess}{\text{$m$}}
\newcommand{\NumOverlapInterval}{\text{$I$}}
\newcommand{\OverlapSet}{\text{${\bf S}$}}

\newcommand{\ParityBit}{\text{$\phi$}}
\newcommand{\SomeInteger}{\text{$N$}}

\newcommand{\FullRecFraction}{\text{$\theta$}}
\newcommand{\FullRecConfidence}{\text{$\gamma$}}
\newcommand{\FullRecProportion}{\text{$\Gamma$}}

\DeclareMathOperator*{\argmin}{argmin}

\ifarxiv
\newcommand{\BEGINPROOF}{\begin{proof}}
\newcommand{\ENDPROOF}{\end{proof}}
\else
\newcommand{\BEGINPROOF}{\begin{IEEEproof}}
\newcommand{\ENDPROOF}{\end{IEEEproof}}
\fi

\newcommand{\JohannesMerkle}{Johannes Merkle\thanks{\ifarxiv\else J. Merkle is with \fi secunet Security Networks, Mergenthaler Allee 77, 65760 Eschborn, Germany. Email: \href{mailto:johannes.merkle@secunet.com}{johannes.merkle@secunet.com}}}
\newcommand{\BenjaminTams}{Benjamin Tams\thanks{\ifarxiv\else B. Tams is with the \fi Institute for Mathematical Stochastics, University of Goettingen, Goldschmidtstr. 7, 37077, Goettingen, Germany. Phone: +49-(0)551-3913515. Email: \href{mailto:btams@math.uni-goettingen.de}{btams@math.uni-goettingen.de}. B. Tams gratefully acknowledges the support of the Felix Bernstein Institute for Mathematical Statistics in the Biosciences and the Volkswagen Foundation.}}

\input{style}

\begin{document}

\title{
Security of the Improved Fuzzy Vault Scheme\\in the Presence of Record Multiplicity\\(Full Version)
}
\author{
\iffull
\JohannesMerkle~~and~\BenjaminTams
\fi
} 
\maketitle

\begin{abstract}
Dodis \etal{} proposed an improved version of the \emph{fuzzy vault scheme}, one of the most popular primitives used in biometric cryptosystems, requiring less storage and leaking less information. Recently, Blanton and Aliasgari have shown that the relation of two improved fuzzy vault records of the same individual may be determined by solving a system of non-linear equations. However, they conjectured that this is feasible for small parameters only. In this paper, we present a new attack against the improved fuzzy vault scheme based on the \emph{extended Euclidean algorithm} that determines if two records are related and recovers the elements by which the protected features, \eg{} the biometric templates, differ. Our theoretical and empirical analysis demonstrates that the attack is very effective and efficient for practical parameters. Furthermore, we show how this attack can be extended to fully recover both feature sets from related vault records much more efficiently than possible by attacking each record individually. We complement this work by deriving lower bounds for record multiplicity attacks and use these to show that our attack is asymptotically optimal in an information theoretic sense. Finally, we propose remedies to harden the scheme against record multiplicity attacks.
\end{abstract}

\ifarxiv
\section*{Keywords}
\else
\begin{keywords}
\fi
fuzzy vault scheme, record multiplicity attack, cross-matching
\ifarxiv
\else
\end{keywords}
\fi

\input{introduction.tex}

\input{scheme.tex}

\input{attack.tex}

\input{bounds.tex}

\input{preventions}

\input{discussion.tex}

\bibliographystyle{IEEEtran}
\bibliography{IEEEabrv,literature}

\iffull
\input{appendix.tex}

\fi

\end{document}

%% file: style.tex
\newif\ifukenglish

\newcommand{\etal}{\textit{et al.}}

\ifukenglish
\newcommand{\analyze}{analyse}

\newcommand{\analyzing}{analysing}

\newcommand{\enrollment}{enrolment}
\newcommand{\enrollments}{enrolments}

\newcommand{\enrolled}{enrolled}

\newcommand{\ie}{\textit{i.e.}}
\newcommand{\eg}{\textit{e.g.}}

\else
\newcommand{\analyze}{analyze}

\newcommand{\analyzing}{analyzing}

\newcommand{\enrollment}{enrollment}
\newcommand{\enrollments}{enrollments}

\newcommand{\enrolled}{enrolled}

\newcommand{\ie}{\textit{i.e.},}
\newcommand{\eg}{\textit{e.g.},}

\fi

%% file: introduction.tex
\section{Introduction}
The \emph{fuzzy vault scheme} by Juels and Sudan \cite{bib:JuelsSudan2002,bib:JuelsSudan2006} is a cryptographic primitive for error-tolerant authentication based on unordered feature sets without revealing the features. It is considered as a potential tool for implementing \emph{biometric cryptosystems} that allow authentication and key derivation based on protected biometric features, \eg{} extracted from a human's fingerprints \cite{bib:ClancyKiyavashLin2003,bib:UludagJain2006,bib:NandakumarJainPankanti2007,bib:Nagar2010,bib:LiEtAl2010}, irises \cite{bib:LeeEtAl2008}, and even face \cite{bib:FrassenZhouBusch2008}. Roughly speaking, it works by hiding the set of \emph{genuine features} in a randomly generated set of \emph{chaff features}, and its security is based on the hardness of the \emph{polynomial reconstruction problem} (see \cite{bib:KiayiasYung2008}). 

However, a serious problem is its vulnerability to \emph{record multiplicity attacks} that link multiple vault records protecting features of the same biometric instance, \eg{} the same finger and uncover the protected biometric data \cite{bib:ScheirerBoult2007, bib:KholmatovYanikoglu2008}: Via correlation, related vault records may be recognized across different databases which conflicts with the \emph{unlinkability} requirement for biometric information protection \cite{bib:ISO24745:2011}; even worse, given records of overlapping feature sets, the common features can be easily recovered, violating the \emph{irreversibility} requirement. 

If the individual feature elements can be robustly represented,\footnote{A nice property of the fuzzy vault scheme is the possibility to tolerate noise not only on the set level but also in the representations of the individual elements, \eg{} introduced by measurement errors.} \eg{} by using sufficiently accurate measurements or applying quantization, correlation attacks can be thwarted by filling up the entire space with chaff points. Unfortunately, this countermeasure drastically inflates the vault records. An alternative is provided by the \emph{improved fuzzy vault scheme} proposed by Dodis \etal{} \cite{bib:DodisEtAl2004,bib:DodisEtAl2008} in which the chaff features are replaced by a polynomial. As a consequence, the data size of the vault records are significantly smaller as compared to those in the original fuzzy vault scheme. Furthermore, since the polynomial represents a maximal number of chaff features, the information leakage becomes minimal and the correlation attacks against the original fuzzy vault do not apply.

In 2013, Blanton and Aliasgari \cite{bib:BlantonAliasgari2013} showed that the improved fuzzy vault scheme is, in principle, also susceptible to record multiplicity attacks, because two vault records of overlapping feature sets leak the elements by which the sets differ. They argued that, for larger feature sets, determining the leaked features may be computationally impossible --- a conjecture disproved by our attacks.

\subsection{Related Work}
There are other schemes that are considered to implement biometric protection. A popular example is the \emph{fuzzy commitment scheme}, which has been proposed in 1999 by Juels and Wattenberg \cite{bib:JuelsWattenberg1999}. More general concepts, called \emph{secure sketches} or \emph{fuzzy sketches}, have been introduced by Dodis \etal{} in 2004 \cite{bib:DodisEtAl2004}. Given multiple instances of a certain fuzzy sketch protecting templates of the same individual, one may ask whether their correspondence can be examined, \ie{} \emph{cross-matching}. 

In 2009, Simoens \etal{} \cite{bib:SimoensEtAl2009} introduced cross-matching attacks as well as attacks via record multiplicity (called \emph{distinguishability attack} and \emph{reversibility attack}, respectively, in that paper); they observed that, for example, the fuzzy commitment scheme may be vulnerable to such kinds of attacks. In 2011, Kelkboom \etal{} \cite{bib:KelkboomEtAl2011} showed how cross-matching performance in a fuzzy commitment scheme is related to the system's authentication performance; furthermore, they proposed to apply (public) random permutation processes to prevent the attacks presented in \cite{bib:SimoensEtAl2009}. Similar ideas have already been briefly noted by Bringer \etal{} in 2008 \cite{bib:BringerEtAl2008} where the use of \emph{cancelable biometrics} was considered to potentially prevent cross-matching attacks. 

In 2012, Blanton and Aliasgari \cite{bib:BlantonAliasgari2012} considered other known fuzzy sketch constructions and concluded that they may be vulnerable to cross-matching and attacks via record multiplicity as well; they also proposed the use of a (secret) key in order to mitigate these vulnerabilities. More details on the approaches of these attacks against fuzzy sketches can be found in \cite{bib:BlantonAliasgari2013}. Blanton and Aliasgari's considerations \cite{bib:BlantonAliasgari2013} also comprise attacks against the improved fuzzy vault scheme. In particular, they observed that, if two instances of the improved fuzzy vault scheme protect sufficiently overlapping feature sets, their set differences, \ie{} the differing elements, can, in principle, be reconstructed by solving a system of non-linear equations. However, the authors admit that their proposed method is computationally feasible only when the number of differing elements is small, which implies that it may not be applicable for practically relevant parameters. 

\subsection{Contribution and Outline of the Paper}

In Section \ref{sec:attack}, we show that an efficient (\ie{} random polynomial-time) record multiplicity attack against the improved fuzzy vault scheme exists. Based on the extended Euclidean algorithm, the attack can reliably determine if two vault records are related, \ie{} if they belong to the same individual, and can partially uncover the features, \eg{} the biometric data. Specifically, if the feature sets protected by the vault records sufficiently overlap, the attack can recover the feature elements in which these sets differ. Thus, this \emph{partial recovery attack} solves the equations established by Blanton and Aliasgari \cite{bib:BlantonAliasgari2013} efficiently even for large parameters. Thereby, we cannot only determine the differing elements from the feature sets but also distinguish related from non-related vault correspondences. We can prove that the attack is always successful provided that the overlap between the feature sets is within the limits determined by Blanton and Aliasgari. In particular, the attack can successfully link two records of the same user with high probability and uncover the feature elements in which the feature sets differ. We perform some experiments to demonstrate the effectiveness of our attack. Finally, we show how our partial recovery attack can be extended to a \emph{full recovery attack} that can completely uncover the feature sets with much higher probability than possible with only one record; thereby, we \analyze{} in which cases the recovery is efficiently possible or even becomes deterministic.

In Section \ref{sec:bounds} we extend the results of \cite{bib:DodisEtAl2004} and derive an upper bound for the information leakage of related vault records. We use this result to derive bounds for the success probability of attacks that aim at determining the feature sets (full recovery) or individual elements thereof (partial recovery) from vault records of two overlapping feature sets. Based on these bounds, we can show that our partial recovery attack is optimal with respect to the number of elements extracted, and that our full recovery attack is, for increasing field size, asymptotically optimal. We are also able to show the impossibility to extract elements from the feature sets' overlap with non-negligible probability for sufficiently large finite fields; as a side-product, this also implies that, for sufficiently large finite fields, no attack can extract feature elements from a single vault record. 

In Section \ref{sec:preventions}, we discuss countermeasures that may prevent successful application of record multiplicity attacks in a possible system incorporating the fuzzy vault scheme. Final discussions are drawn in Section \ref{sec:discussion}. 

Before we start with describing our attack, we first describe the functioning and some details on fuzzy vault schemes, original and improved, and examine the differences between them.

%% file: scheme.tex
\section{The Fuzzy Vault Schemes}\label{sec:scheme}

\subsection{The Original Fuzzy Vault Scheme}
In the fuzzy vault vault scheme by Juels and Sudan \cite{bib:JuelsSudan2006}, it is assumed that the features measured from an individual are encoded as a \emph{feature set} $\VeatureSet$ of size $\VeatureSize$ with elements in a fixed finite field $\FiniteField$ of size $\FieldSize=|\FiniteField|$. 

On \enrollment{}, first, a random secret polynomial $\VaultSecret\in\FiniteField[\SYM]$ of degree $<\SecretSize$ is generated; then, the secret polynomial is evaluated on the feature set to build the set of \emph{genuine pairs} $\GenuinePairs\subset\FiniteField\times\FiniteField$ containing the pairs $(\absc,\VaultSecret(\absc))$ where $\absc\in\VeatureSet$; note that, if $\SecretSize\leq\GenuineSize$, then $\GenuinePairs$ uniquely encodes $\VaultSecret$ and the features $\VeatureSet$; to hide both $\VaultSecret$ and $\VeatureSet$, a large set $\ChaffPairs$ of random chaff pairs $(\absc,\ord)$ where $\absc\notin\VeatureSet$ and $\ord\neq\VaultSecret(\absc)$ is generated; finally, the union $\VaultPairs=\GenuinePairs\cup\ChaffPairs$ of size $\VaultSize$ builds the vault.

On authentication, a second query feature set $\WeatureSet$ of the allegedly genuine individual is provided; it is used to extract the set of \emph{unlocking pairs} $\UnlockingPairs$ containing those vault pairs whose abscissa, \ie{} \emph{vault features}, well approximate the elements encoded by $\WeatureSet$; note that, since the elements in $\VeatureSet$ and $\WeatureSet$ encode biometric features (\eg{} minutiae), the similarity measure between them (\eg{} minutiae distance) can be used to determine which query features well approximate vault features; if the unlocking pairs $\UnlockingPairs$ dominantly contains genuine pairs, \ie{} pairs lying on the graph of a common polynomial of degree smaller than $\SecretSize$, the polynomial $\VaultSecret$ can be recovered which deems to be a match. In particular, if $\UnlockingPairs$ contains at least $(|\UnlockingPairs|+\SecretSize)/2$ genuine pairs, an algorithm for decoding Reed-Solomon codes can be used to recover $\VaultSecret$ efficiently (\eg{} see \cite{bib:Berlekamp1984,bib:Gao2002}).

An adversary having intercepted a vault $\VaultPairs$ (without knowing a matching template), has to identify at least $\SecretSize$ genuine pairs in $\VaultPairs$. From the difficulty in identifying $\SecretSize$ genuine pairs, the fuzzy vault scheme achieves its security. This can be either ensured information-theoretically for suitable $|\FiniteField|$, $\SecretSize$, $\GenuineSize$, and $\VaultSize$ or, which is currently of greater practical relevance, using the widely believed computational hardness of the polynomial reconstruction problem (see \cite{bib:GuruswamiSudan1998,bib:BleichenbacherNguyen2000,bib:GuruswamiVardy2005}). 

\subsection{The Improved Fuzzy Vault Scheme}
Similar as in the original scheme, for the improved fuzzy vault scheme, it is assumed that the biometric features measured from an individual are encoded as a feature set $\VeatureSet$ of size $\GenuineSize$ containing elements in a fixed finite field $\FiniteField$. However, unlike the original fuzzy vault scheme, the improved fuzzy vault does not explicitly store the reference features and, thus, does not allow to account for small deviations in the measurement of the individual features by comparing the features in the query set with the reference features. As a consequence, the encoding of the features to field elements must be robust w.r.t. typical measurement errors, and two features are considered to match if and only if they have an equal encoding in $\FiniteField$.

Dodis \etal{} proposed two variants of their improved fuzzy vault scheme: In \cite{bib:DodisEtAl2004}, a probabilistic version was proposed, in which, similar to the original fuzzy vault scheme, a secret polynomial is randomly chosen and used to hide the genuine features. In contrast, the construction proposed in \cite{bib:DodisEtAl2008} is deterministic. With respect to security against recovery of the features from the vaults, both variants are equivalent, and in our analysis, we will consider both of them.

\iffull
\subsubsection{Probabilistic Variant}
\fi

On \enrollment{}, as in the original vault, a random secret polynomial $\VaultSecret\in\FiniteField[\SYM]$ of degree smaller than $\SecretSize$ is generated; now, the secret polynomial is bound to the feature set $\VeatureSet$ by publishing the polynomial
\begin{equation*}
\VaultPoly(\SYM)=\VaultSecret(\SYM)+\prod_{\absc\in\VeatureSet}(\SYM-\absc)
\end{equation*}
which is monic and of degree $\GenuineSize=|\VeatureSet|$. We denote $\VaultPoly$ as a {\em vault record} of the feature set $\VeatureSet$.

Note that, if $\absc\in\VeatureSet$, then $\VaultPoly(\absc)=\VaultSecret(\absc)$ and thus $(\absc,\VaultPoly(\absc))$ is a genuine pair; otherwise, if $\absc\notin\VeatureSet$, then $\VaultPoly(\absc)\neq\VaultSecret(\absc)$ and $(\absc,\VaultPoly(\absc))$ is a chaff pair.

On authentication, given the query feature set $\WeatureSet$, the unlocking set $\UnlockingPairs$ is computed containing the pairs $(\absc,\VaultPoly(\absc))$ where $\absc\in\WeatureSet$; the remaining, \ie{} the attempt for decoding $\UnlockingPairs$, \eg{} using an algorithm for decoding Reed-Solomon codes, works analogous as on authentication in the original fuzzy vault scheme.

Obviously, the secret polynomial $\VaultSecret$ serves to protect the feature set. If $\VaultSecret$ were not added to the characteristic polynomial 
\begin{equation*}
\CharPoly_{\VeatureSet}(\SYM)=\prod_{\absc\in\VeatureSet}(\SYM-\absc)=\VeatureCoeff_0+\VeatureCoeff_1\SYM+\hdots+\VeatureCoeff_{\GenuineSize-1}\SYM^{\GenuineSize-1}+\SYM^\GenuineSize
\end{equation*}
of the feature set $\VeatureSet$, the features could be efficiently recovered using an efficient polynomial-time root-finding algorithm (\eg{} via polynomial factorization; see \cite{bib:vzGathen2003}). However, the recovery of the feature set $\VeatureSet$ is hardened by blinding $\VeatureCoeff_0,\ldots, \VeatureCoeff_{\SecretSize-1}$ with $\VaultSecret$. Due to this randomization, the lowest $\SecretSize$ coefficients of $\VaultPoly$ do not carry any information and, hence, could even be dismissed without affecting the protection of the feature set $\VeatureSet$. This approach results in the deterministic variant of the improved fuzzy vault scheme \cite{bib:DodisEtAl2008} when selecting $\VaultSecret(\SYM)=-(\VeatureCoeff_0+\VeatureCoeff\SYM+\hdots+\VeatureCoeff_{\SecretSize}\SYM^{\SecretSize-1})$.

\iffull
\subsubsection{Deterministic Variant}
The deterministic version proposed in \cite{bib:DodisEtAl2008} is a slight modification of the probabilistic construction of \cite{bib:DodisEtAl2004} presented in the previous section. Instead of blending the lowest $\SecretSize$ coefficients $\VeatureCoeff_{0},\ldots,\VeatureCoeff_{\SecretSize-1}$ of the characteristic polynomial $\CharPoly_{\VeatureSet}$ with a random polynomial $\VaultSecret$ of degree smaller than $\SecretSize$, these coefficients are dismissed and just $\VeatureCoeff_{\SecretSize}\SYM^{\SecretSize} + \cdots + \VeatureCoeff_{\GenuineSize-1}\SYM^{\GenuineSize-1}+\SYM^\GenuineSize$ as a part of $\CharPoly_{\VeatureSet}$ is stored as the vault. Note that, in the deterministic version, the size of the vault record is reduced by $\SecretSize$ finite field elements, \ie{} $\SecretSize\cdot\log\FieldSize$ bits where $\FieldSize=|\FiniteField|$.

In order to facilitate our analysis, we will consider both variants as special cases of a {\em generalized improved fuzzy vault scheme}, where the vault record $\VaultPoly$ is set to $\VaultPoly=\VaultSecret+\CharPoly_{\VeatureSet}$ with a polynomial $\VaultSecret$ of degree smaller $\SecretSize$. In case that $\VaultSecret$ is chosen at random, we obtain the probabilistic version, and if we set $\VaultSecret(\SYM)=-(\VeatureCoeff_0+\VeatureCoeff\SYM+\hdots+\VeatureCoeff_{\SecretSize}\SYM^{\SecretSize-1})$ we arrive at the deterministic version. 
\fi

\subsection{Security of Single Vault Records} \label{security_single}
In \cite{bib:DodisEtAl2004} and \cite{bib:DodisEtAl2008},  Dodis \etal{} presented information theoretic results on the security of both versions. In particular, they used the \emph{average min-entropy} to measure the information leaked by the vault records.
\begin{defn}
The \emph{min-entropy} of a variable $\RandVarX$ is defined as 
\begin{equation*}
\MinEntropy (\RandVarX) = -\log \left( \max_{\varx}( \Pr [\RandVarX = \varx] ) \right),
\end{equation*}  
and the \emph{average min-entropy} of $\RandVarX$ given $\RandVarY$ is defined as 
\begin{equation*}
\AverageMinEntropy (\RandVarX | \RandVarY) = - \log \left( \Expectation_{y \leftarrow\RandVarY} \left[  \max_{\varx}( \Pr [\RandVarX = \varx | \RandVarY=\vary]) \right] \right),
\end{equation*}
where $\Expectation_{\vary \leftarrow \RandVarY}$ denotes expectation for $\vary$ chosen at random from $\RandVarY$. The difference $\MinEntropy (\RandVarX)-\AverageMinEntropy (\RandVarX | \RandVarY)$ is called the {\em information leakage} or {\em entropy loss} of $\RandVarX$ by $\RandVarY$.
\end{defn}
As shown in \cite{bib:DodisEtAl2004} and \cite{bib:DodisEtAl2008}, the average min-entropy $\AverageMinEntropy (\VeatureSet  | \VaultPoly)$ is at least $\MinEntropy (\VeatureSet)- (\GenuineSize - \SecretSize)\log \FieldSize$. This gives an upper bound of $\FieldSize^{-\MinEntropy (\VeatureSet) + \GenuineSize - \SecretSize}$ for the average success probability of an attacker who tries to determine the feature set $\VeatureSet$ from the vault record $\VaultPoly$, where the average is taken over the randomness in $\VaultPoly$. However, even in the probabilistic version of the scheme, all information leakage results from the deterministic coefficients in the vault record. Therefore, the bound on the success probability also holds for any given vault record $\VaultPoly$. 

In the case where all feature elements are independently and uniformly chosen at random, the feature set $\VeatureSet$ has maximal entropy $\log \tbinom{\FieldSize}{\GenuineSize}$. Therefore, we can conclude that, unless an attacker can exploit a potential non-uniformity and interdependency of the features, his success probability $\prob{Single}$ is limited by 
\begin{equation}\label{max_prob}
\prob{Single} \leq \FieldSize^{\GenuineSize - \SecretSize} / \tbinom{\FieldSize}{\GenuineSize}. 
\end{equation}
For $\FieldSize \gg \GenuineSize$, this can be approximated by $\GenuineSize ! / \FieldSize^{\SecretSize}$. On the other hand, the best known attack method is to guess $\SecretSize$ feature elements and applying the unlocking process of the scheme which is of success probability $ \tbinom{\GenuineSize}{\SecretSize} / \tbinom{\FieldSize}{\SecretSize}$. For $\FieldSize \gg \GenuineSize$, this can be approximated by $\tbinom{\GenuineSize}{\SecretSize} / \FieldSize^{\SecretSize}$. Thus, for $\FieldSize \gg \GenuineSize$, the bound on the average min-entropy is tight up to a factor of at most $(\GenuineSize-\SecretSize)! \, \SecretSize !$. 

It is important to note that in certain applications of the scheme, it may be possible to exploit a low entropy of the features. For instance, if the scheme is used to protect biometric features, an attacker can randomly choose feature sets from a large database of biometric data. As shown in \cite{bib:Tams2013v1}, such a {\em false-accept attack} can be quite efficient under practical circumstances.

%% file: attack.tex
\section{Attacks}\label{sec:attack}
Subsequently, assume that an adversary is given two vaults
\begin{align}\label{eq:vaults}
\begin{split}
\VaultPoly(\SYM)&=\VaultSecret(\SYM)+\prod_{\VeatureElement_i\in\VeatureSet}(\SYM-\VeatureElement_i)\quad\text{and}\\
\WaultPoly(\SYM)&=\WaultSecret(\SYM)+\prod_{\WeatureElement_i\in\WeatureSet}(\SYM-\WeatureElement_i)
\end{split}
\end{align}
that protect the feature sets $\VeatureSet$ and $\WeatureSet$ of size $\VeatureSize$ and $\WeatureSize\leq \VeatureSize$, respectively, overlapping in $\NumOverlap=|\VeatureSet\cap\WeatureSet|$ elements; $\VaultSecret$ and $\WaultSecret$ are (yet unknown) polynomials in $\FiniteField[\SYM]$ of degree smaller $\SecretSize$. 

\input{partialrec}
\input{experiments}
\input{fullrec}

%% file: partialrec.tex
\subsection{The Approach of Blanton and Aliasgari}
Blanton and Aliasgari \cite{bib:BlantonAliasgari2013} observed that, if $\VeatureSize=\WeatureSize$ and the feature sets $\VeatureSet$ and $\WeatureSet$ overlap in $\NumOverlap$ elements, a system of polynomial equations in the $2(\VeatureSize-\NumOverlap)$ distinct elements can be derived from the vault records. We give a brief summary of their approach.

Using (\ref{eq:vaults}), we can represent the coefficients $\VaultCoeff_\SecretSize,\ldots,\VaultCoeff_{\VeatureSize}$ and $\WaultCoeff_\SecretSize,\ldots, \WaultCoeff_{\VeatureSize}$ of the vault records $\VaultPoly$ and $\WaultPoly$, respectively, as $\VaultCoeff_j=\sigma_{\VeatureSize - j}(\VeatureElement_1,\ldots, \VeatureElement_{\VeatureSize})$ for $j=\SecretSize,\ldots,\VeatureSize$ and $\WaultCoeff_j=\sigma_{\VeatureSize - j}(\WeatureElement_1,\ldots, \WeatureElement_{\VeatureSize})$ for $j=\SecretSize,\ldots,\VeatureSize$. 
where 
\begin{equation}\label{eq:elsympol}
\sigma_{m}(\SYM_1,\ldots, \SYM_N) = \sum_{\substack{J \subseteq \{1, \ldots, N\} \\ |J|=m}} \prod_{j\in J} \SYM_{j}
\end{equation}
denotes the $m$-th elementary symmetric polynomial in $N$ variables \cite{bib:LidlNiederreiter} (note, that $\sigma_{0}(\SYM_1,\ldots, \SYM_N)=1$.) Assume that $\VeatureElement_i=\WeatureElement_i$ for $i=1,\ldots \NumOverlap$. Using the identity 
\ifdouble
\begin{align}
\begin{split}
\label{eq:esp_identity}
\sigma_{m}(\SYM_1,\ldots, \SYM_N)& \\
=\sum_{i=0}^{m} &\sigma_{i}(\SYM_1,\ldots, \SYM_M) \sigma_{m-i}(\SYM_{M+1},\ldots, \SYM_N),  
\end{split}
\end{align}
\else
\begin{equation}\label{eq:esp_identity}
\sigma_{m}(\SYM_1,\ldots, \SYM_N) = \sum_{i=0}^{m} \sigma_{i}(\SYM_1,\ldots, \SYM_M) \sigma_{m-i}(\SYM_{M+1},\ldots, \SYM_N),  
\end{equation}
\fi
for some $M<N$, we obtain 
\ifdouble
\begin{align}
\begin{split}
\VaultCoeff_{t-j}& - \WaultCoeff_{t-j}\\
&=\sigma_{j}(\VeatureElement_1,\ldots, \VeatureElement_{\VeatureSize})  - \sigma_{j}(\VeatureElement_1,\ldots,\VeatureElement_{\NumOverlap}, \WeatureElement_{\NumOverlap+1},\ldots, \WeatureElement_{\VeatureSize}) \\ 
&=\sum_{i=0}^{j-1}\sigma_{i}(\VeatureElement_1,\ldots,\VeatureElement_{\NumOverlap}) \cdot \\
 &\quad\quad~\left( \sigma_{j-i}(\VeatureElement_{\NumOverlap+1},\ldots, \VeatureElement_{\VeatureSize}) - \sigma_{j-i}(\WeatureElement_{\NumOverlap+1},\ldots, \WeatureElement_{\VeatureSize}) \right),
\end{split}\label{eq:blanton}
\end{align}
\else
\begin{align}
\begin{split}
\VaultCoeff_{t-j} &- \WaultCoeff_{t-j}  \\
 & = \sigma_{j}(\VeatureElement_1,\ldots, \VeatureElement_{\VeatureSize})  - \sigma_{j}(\VeatureElement_1,\ldots,\VeatureElement_{\NumOverlap}, \WeatureElement_{\NumOverlap+1},\ldots, \WeatureElement_{\VeatureSize}) \\ 
&= \sum_{i=0}^{j-1} \sigma_{i}(\VeatureElement_1,\ldots,\VeatureElement_{\NumOverlap}) \cdot 
 \left( \sigma_{j-i}(\VeatureElement_{\NumOverlap+1},\ldots, \VeatureElement_{\VeatureSize}) - \sigma_{j-i}(\WeatureElement_{\NumOverlap+1},\ldots, \WeatureElement_{\VeatureSize}) \right),
\end{split}\label{eq:blanton}
\end{align}
\fi
for $j=1,\ldots,\VeatureSize-\SecretSize$. By successively inserting 
\ifdouble
\begin{align*}
\begin{split}
\sigma_{i}(\VeatureElement_1,\ldots,\VeatureElement_{\NumOverlap})&\\
=~\VaultCoeff_{t-i}-&\sum_{l=0}^{i-1} \sigma_{l}(\VeatureElement_1,\ldots,\VeatureElement_{\NumOverlap}) \cdot \sigma_{i-l}(\VeatureElement_{\NumOverlap+1},\ldots, \VeatureElement_{\VeatureSize}), 
\end{split}
\end{align*}
\else
\begin{equation*}
\sigma_{i}(\VeatureElement_1,\ldots,\VeatureElement_{\NumOverlap})=\VaultCoeff_{t-i}-\sum_{l=0}^{i-1} \sigma_{l}(\VeatureElement_1,\ldots,\VeatureElement_{\NumOverlap}) \cdot \sigma_{i-l}(\VeatureElement_{\NumOverlap+1},\ldots, \VeatureElement_{\VeatureSize}), 
\end{equation*}
\fi
(which follows from (\ref{eq:esp_identity}) as well) for $i=j-1,\ldots, 1$, we can clear (\ref{eq:blanton}) from all terms depending on merely $\VeatureElement_1,\ldots,\VeatureElement_{\NumOverlap}$, resulting in a system of $\VeatureSize-\SecretSize$ polynomial equations in the unknowns $\VeatureElement_{\NumOverlap+1},\ldots,\VeatureElement_{\VeatureSize},\WeatureElement_{\NumOverlap+1},\ldots,\WeatureElement_{\VeatureSize}$. If $\NumOverlap \geq (\VeatureSize+\SecretSize)/2$, this system of non-linear equations can be solved. 

In this approach, two aspects remained open:
\begin{enumerate}
\item {\bf Can the equations be solved efficiently?} Blanton and Aliasgari suggested standard methods for solving polynomial equations and assumed that this task becomes infeasible for large $\VeatureSize$ and $\FieldSize$. We disprove this conjecture by our partial recovery attack described below.
\item {\bf Do the solutions found indeed correspond to the correct feature set?} Since the equations are of total degree up to $\VeatureSize-\SecretSize$, while there are only $2(\VeatureSize-\NumOverlap)!$ many permutations of the correct feature elements, it is not clear whether there are many spurious solutions.   
\end{enumerate}
Our attack and its subsequent analysis give positive answers to these questions.
 
\subsection{Partial Recovery Attack}\label{sec:PartialRecoveryAttack}

Our partial recovery attack via record multiplicity is given by the algorithm below. 

\begin{alg}[Partial Recovery Attack]
\label{Algo:partial_recovery}
\AlgInput Two vault records $\VaultPoly$, $\WaultPoly$ of sets $\VeatureSet$ and $\WeatureSet$ of size $\VeatureSize$ and $\WeatureSize\leq\VeatureSize$, respectively.
\AlgOutput Either a triple $(\NumOverlap^*,\VeatureSet_0,\WeatureSet_0)$, where $\NumOverlap^*$ is a candidate for $|\VeatureSet\cap\WeatureSet|$, $\VeatureSet_0$ is a candidate for $\VeatureSet\setminus\WeatureSet$, and $\WeatureSet_0$ is a candidate for $\WeatureSet\setminus\VeatureSet$, or \AlgFailure.
\end{alg}
\begin{enumerate}
\item\label{partial_eea} Apply the extended Euclidean algorithm to $\VaultPoly$ and $\WaultPoly$ to obtain a list of combinations $\RemPoly_j=\VErrorPoly_j\cdot\VaultPoly+\WErrorPoly_j\cdot\WaultPoly$.
\item\label{partial_index} Let $j_0$ be such that $\deg(\WErrorPoly_{j_0})$ is minimal where $\deg(\WErrorPoly_{j_0})+\SecretSize>\deg(\RemPoly_{j_0})$ and $\RemPoly_{j_0}\neq 0$; if such an index does not exist, return \AlgFailure.
\item\label{partial_rem} Ensure that the degree of the remainder of $\VaultPoly$ divided by $\WErrorPoly_{j_0}$ is smaller than $\SecretSize$; otherwise, if it is greater than or equal $\SecretSize$, return $\AlgFailure$.
\item\label{partial_roots} Compute the roots $\VeatureSet_0$ and $\WeatureSet_0$ of $\WErrorPoly_{j_0}$ and $\VErrorPoly_{j_0}$, respectively; if $\WErrorPoly_{j_0}$ or $\VErrorPoly_{j_0}$ do not split into distinct linear factors, then return \AlgFailure; otherwise set $\NumOverlap=\VeatureSize-\deg(\WErrorPoly_{j_0})$ and return the triple $(\NumOverlap,\VeatureSet_0,\WeatureSet_0)$.
\end{enumerate}

In the following, we will state provable results in which case the attack's output is correct. 


The result of our analysis of Algorithm~\ref{Algo:partial_recovery} is summarized in the following. Its proof, which in particular uses a result from Gao \cite{bib:Gao2002}, 
\iffull
can be found in the appendix in Section \ref{sec:ProofMain}.
\fi

\begin{thm}\label{thm:main}
Let $\VeatureSet,\WeatureSet\subset\FiniteField$ be feature sets of size $\VeatureSize$ and $\WeatureSize\leq\VeatureSize$, respectively, and let $\VaultSecret,\WaultSecret\in\FiniteField[\SYM]$ be of degree smaller than $\SecretSize$. Furthermore, let $\NumOverlap=|\VeatureSet\cap\WeatureSet|$.
\begin{itemize}
\item[a)] If the vaults $\VaultPoly=\VaultSecret+\CharPoly_\VeatureSet$ and $\WaultPoly=\WaultSecret+\CharPoly_\WeatureSet$ are input to Algorithm~\ref{Algo:partial_recovery} where $\NumOverlap\geq(\VeatureSize+\SecretSize)/2$, then the algorithm outputs $(\NumOverlap,\VeatureSet\setminus\WeatureSet,\WeatureSet\setminus\VeatureSet)$.
\item[b)]Suppose that for normalized polynomials $\VaultPoly$ and $\WaultPoly$ of degree $\VeatureSize$ and $\WeatureSize$, respectively, Algorithm \ref{Algo:partial_recovery} outputs $(\NumOverlap^*,\VeatureSet_0,\WeatureSet_0)$. Then there exists polynomials $\hat{\VaultSecret},\hat{\WaultSecret}\in\FiniteField[\SYM]$ of degree smaller than $\SecretSize$ and a polynomial $\CharPoly\in\FiniteField[\SYM]$ of degree $\NumOverlap^*$ such that $\VaultPoly=\hat{\VaultSecret}+\CharPoly\cdot\CharPoly_{\VeatureSet_0}$ and $\WaultPoly=\hat{\WaultSecret}+\CharPoly\cdot\CharPoly_{\WeatureSet_0}$.
\item[c)]Algorithm~\ref{Algo:partial_recovery} can be implemented using an expected number of $\Oh(\VeatureSize^2)+\SoftOh(\VeatureSize\cdot\log|\FiniteField|)$ operations in $\FiniteField$.
\end{itemize}
\end{thm}
Note that, if the system uses a Reed-Solomon decoder, successful authentication requires the query feature set and the \enrolled{} feature sets to share at least $(\VeatureSize+\SecretSize)/2$ elements. Consequently, the probability that a query set and an \enrolled{} set overlap in $(\VeatureSize+\SecretSize)/2$ elements equals the system's {\em genuine acceptance rate} (GAR). Since, typically, during \enrollment{}, features are measured with at least the same accuracy as during verification, we can expect the average overlap between two \enrolled{} sets to be at least as high as between query sets and \enrolled{} sets. Therefore, by Statement a) in Theorem \ref {thm:main}, we can expect that Algorithm~\ref{Algo:partial_recovery} successfully links two vault records of the same individual, and uncovers the differences between the two \enrolled{} feature sets with a probability at least equal to the GAR. 

If Algorithm~\ref{Algo:partial_recovery} outputs a triple, then, by Statement b), we cannot necessarily assume that the result corresponds to feature sets protected by the vault. The attack might output a triple $(\NumOverlap^*,\VeatureSet_0,\WeatureSet_0)$ even if the input vaults protect feature sets sharing less than $(\VeatureSize+\SecretSize)/2$ elements. There are three possible cases: First, the output equals $(\NumOverlap,\VeatureSet\setminus\WeatureSet,\WeatureSet\setminus\VeatureSet)$ (in particular, $\NumOverlap<(\VeatureSize+\WeatureSize)/2$); second, there exists other feature sets $\VeatureSet',\WeatureSet'$ in $\FiniteField$ with $\NumOverlap^*=|\VeatureSet'\cap\WeatureSet'|$, $\VeatureSet_0=\VeatureSet'\setminus\WeatureSet'$, and $\WeatureSet_0=\WeatureSet'\setminus\VeatureSet'$; third, the polynomial $\CharPoly$ of Statement b) does not split into linear factors in $\FiniteField$ or has multiple roots. If the first case occurs, we call the output \emph{correct}; otherwise, we call the output \emph{spurious}. 
\iffull
We will see later by experiments (Section \ref{sec:PartRecsExtField}) that the third case definitely occurs  even if $\NumOverlap^*\geq(\VeatureSize+\SecretSize)/2$ is output.
\fi

As spurious outputs do not reveal correct information, their occurrence are inadvertent cases for an adversary who is using the partial recovery attack to link vault records across different databases (or even break them). Therefore, if the probability of spurious outputs is reasonably high, the adversary might not gain much since he cannot necessarily distinguish correct from spurious outputs. However, in the next section we show by experiments that the likelihood of spurious outputs can be very small in practical circumstances. And even if spurious outputs occur frequently, an adversary having intercepted two vaults from which he knows that they are related, can hope that they sufficiently overlap and try to discover their differences explicitly via the partial recovery attack, which also eases subsequent attacks to fully recovery the features.

%% file: experiments.tex
\subsection{Experiments}\label{sec:experiments}
We conducted experiments with an implementation of Algorithm \ref{Algo:partial_recovery} for various parameters comprising configurations that we expect to encounter in practice.\footnote{Our experiments can be reproduced using a C++ library THIMBLE of which source can be downloaded from \url{http://www.stochastik.math.uni-goettingen.de/biometrics/thimble}. Note that a tutorial on how to run the attack is contained in the \href{http://www.stochastik.math.uni-goettingen.de/biometrics/fileadmin/thimble/doc-2014.09.07/EEAAttack_8h.html\#sec_arm_improved_fv}{documentation}.} In a nutshell, we found that the partial recovery attack is computationally very efficient and outputs the differing elements if the feature sets overlap in $\NumOverlap\geq (\VeatureSize+\SecretSize)/2$ elements, and succeeds even for slightly smaller overlaps with non-negligible probability. Furthermore, we observed that, for $\NumOverlap< (\VeatureSize+\SecretSize)/2$ and typical parameters suggested for implementations of the (original) fuzzy vault, the likelihood of spurious outputs is very small. Consequently, we conclude that the partial recovery attack is a very efficient tool for an adversary who is attempting to find (and uncover elements from) related vault records from different application's databases.

\subsubsection{Partial Recovery Rates}

For the ease of reading, we use the following definitions. For a set $\NumOverlapInterval$ of non-negative integers, \eg{} an interval, let $\prob{out}(\NumOverlapInterval)$ and $\prob{cor}(\NumOverlapInterval)$ be the average probabilities among all $\NumOverlap\in\NumOverlapInterval$ that Algorithm \ref{Algo:partial_recovery} outputs a triple and that it outputs a correct triple, respectively; furthermore, let $\prob{out}(\NumOverlap)=\prob{out}(\{\NumOverlap\})$ and $\prob{cor}(\NumOverlap)=\prob{cor}(\{\NumOverlap\})$.

For each $(\FieldSize,\VeatureSize,\SecretSize,\NumOverlap)$ with $\FieldSize=2^8,2^9,\ldots,2^{16}$, $\VeatureSize=\WeatureSize=24,38,44$, $\SecretSize=2,\ldots,\VeatureSize-1$ and $\NumOverlap=0,\ldots,\VeatureSize$ we executed $10^5$ tests, where we randomly generated two feature sets $\VeatureSet,\WeatureSet\subset\FiniteField$ with $\NumOverlap=|\VeatureSet\cap\WeatureSet|$. Two random polynomials $\VaultSecret,\WaultSecret\in\FiniteField[\SYM]$ of degree exactly $\SecretSize-1$ have been generated and then the polynomials $\VaultPoly=\VaultSecret+\CharPoly_\VeatureSet$ and $\WaultPoly=\WaultSecret+\CharPoly_\WeatureSet$ were input to our implementation of our partial recovery attack (Algorithm \ref{Algo:partial_recovery}).

As predicted by Theorem \ref{thm:main}, the attack was always successful for $\NumOverlap\geq (\VeatureSize+\SecretSize)/2$. 

\newcommand{\figax}{0.55}
\newcommand{\figay}{0.95}
\newcommand{\figbx}{0.975}
\newcommand{\figby}{0.75}
\ifsingle
\renewcommand{\figax}{0.4}
\renewcommand{\figbx}{0.75}
\renewcommand{\figby}{0.95}
\fi
\ifdouble
\renewcommand{\figax}{0.35}
\renewcommand{\figbx}{0.6}
\renewcommand{\figby}{0.95}
\fi

\ifdouble
\begin{figure*}[t]
\else
\begin{figure}[!t]
\fi
\begin{center}
\subfigure[Plot of $\prob{cor}(\lceil(\VeatureSize+\SecretSize)/2-1\rceil)$ for varying $\SecretSize$ and different $\FieldSize$ at $\VeatureSize=24$.]{\label{fig:pcorjumps}\begin{tikzpicture}
\begin{axis}[
legend style={at={(\figax,\figay)}},
width=0.49\textwidth,
xmin=1,xmax=24,
ymin=0,ymax=0.013,
enlargelimits=false,
xtick={1,6,11,16,21},
yticklabels={$0.25\%$,$0.50\%$,$0.75\%$,$1.00\%$,$1.25\%$},
ytick={0.0025,0.005,0.0075,0.01,0.0125},
scaled ticks=false,
xlabel={$\SecretSize$},
grid=major
]
\addplot+[mark options={fill=white,draw=blue}] table[x={k},y={pcor8}]{plot24.dat}; \addlegendentry{$\FieldSize=2^8$}
\addplot+[mark=x] table[x={k},y={pcor9}]{plot24.dat}; \addlegendentry{$\FieldSize=2^{9}$}
\addplot+[mark=+] table[x={k},y={pcor10}]{plot24.dat}; \addlegendentry{$\FieldSize=2^{10}$}
\end{axis}
\end{tikzpicture}}
\hspace{0.02\textwidth}
\subfigure[Plot of $\prob{cor}(\lceil(\VeatureSize+\SecretSize)/2-1\rceil)$ and $\prob{out}(\lceil(\VeatureSize+\SecretSize)/2-1\rceil)$ for varying $\SecretSize$ at $\VeatureSize=24$ and $\FieldSize=2^8$.]{\label{fig:pcor_for_varying_secret_sizes}\begin{tikzpicture}
\begin{axis}[
legend style={at={(\figbx,\figby)}},
width=0.49\textwidth,
xmin=2,xmax=19,
ymin=0,ymax=0.03,
enlargelimits=false,
xtick={3,8,13,18},
yticklabels={$0.5\%$,$1\%$,$1.5\%$,$2\%$,$2.5\%$},
ytick={0.005,0.01,0.015,0.02,0.025},
scaled ticks=false,
xlabel={$\SecretSize$},
grid=major
]
\addplot+[mark options={fill=white,draw=blue}] table[x={k},y={pcor8}]{plot24.dat}; \addlegendentry{$\prob{cor}(\lceil(\VeatureSize+\SecretSize)/2-1\rceil)$}
\addplot+[mark=x] table[x={k},y={poutcor8},only marks]{plot24.dat}; \addlegendentry{$\prob{out}(\lceil(\VeatureSize+\SecretSize)/2-1\rceil)$}
\end{axis}
\end{tikzpicture}}
\caption{Plots for the frequencies at which Algorithm \ref{Algo:partial_recovery} outputs a triple and at which an output triple is correct for $\NumOverlap=\lceil(\VeatureSize+\SecretSize)/2-1\rceil$.}
\end{center}
\ifdouble
\end{figure*}
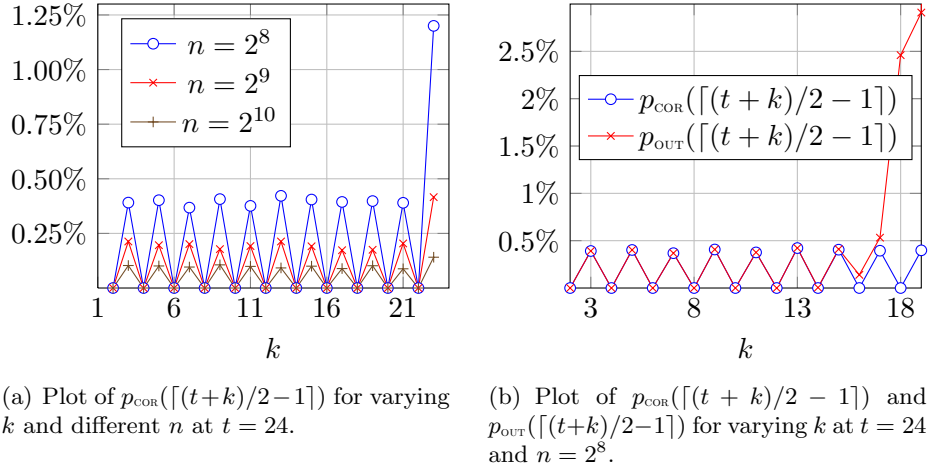
\else
\end{figure}
\fi

For $\NumOverlap<(\VeatureSize+\SecretSize)/2$, we observed correct outputs only when $\VeatureSize+\SecretSize$ was odd and $\NumOverlap$ was maximal, \ie{} $\NumOverlap=\lceil(\VeatureSize+\SecretSize)/2-1\rceil$ (see Figure \ref{fig:pcorjumps}), which is equivalent to $2\NumOverlap=\VeatureSize+\SecretSize-1$. Our experiments show, for $2\NumOverlap=\VeatureSize+\SecretSize-1$ and $\SecretSize$ not too close to $\VeatureSize$, that $\prob{cor}(\NumOverlap)\approx\prob{out}(\NumOverlap) \approx 1/\FieldSize$, \ie{} almost all outputs are correct and the probability of a (correct) output is inversely proportional to the size of the finite field (see Figures \ref{fig:pcorjumps} and \ref{fig:pcor_for_varying_secret_sizes}). 
\iffull
For example, for $\VeatureSize=38$ and $\FieldSize=2^{16}$, the first $\SecretSize$ for which we observed a spurious output was $\SecretSize=26$. 
\fi
As we show in Section \ref{sec:bounds_partial}, the success probability $1/\FieldSize$ for $2\NumOverlap=\VeatureSize+\SecretSize-1$ is asymptotically optimal up to a constant for fixed $\VeatureSize$ and $\FieldSize \rightarrow \infty$. The same success probability is achieved by an attack that, after guessing one element from $(\VeatureSet\cup\WeatureSet)\setminus (\VeatureSet\cap\WeatureSet)$, solves the equations of Blanton and Aliasgari (see Section \ref{sec:PartialRecoveryAttack}); however, Algorithm \ref{Algo:partial_recovery}, beside being computationally more efficient (see below), has the advantage that, unless $\SecretSize$ is close to $\VeatureSize$, due to $\prob{cor}(\NumOverlap)\approx\prob{out}(\NumOverlap)$ there is high assurance that the output is indeed correct.

When $\SecretSize$ approaches $\VeatureSize$, the probability of spurious (incorrect) output triples increases drastically, and, thus, most output triples become spurious, but, yet, the probability $\prob{cor}(\lceil(\VeatureSize+\SecretSize)/2-1\rceil)$ of a correct output is at least  
\iffull
$1/\FieldSize$. Precisely, $\prob{cor}(\lceil(\VeatureSize+\SecretSize)/2-1\rceil)$ approximates $1/\FieldSize$ for $\SecretSize=\VeatureSize-3$ and is slightly higher and dependent on $\VeatureSize$ for $\SecretSize=\VeatureSize-1$ (Figure \ref{fig:pcor_for_varying_field_size}).
\fi

We summarize that the partial recovery attack is a serious threat even for $2\NumOverlap=\VeatureSize+\SecretSize-1$, unless extremely large finite fields are used (which would render the scheme impractical).  

\iffull
\ifdouble
\begin{figure*}[t]
\else
\begin{figure}[!t]
\fi
\begin{center}
\subfigure[For $\SecretSize=\VeatureSize-3$, the average likelihood $\prob{cor}(\lceil(\VeatureSize+\SecretSize)/2-1\rceil)$ plotted versus the field size $\FieldSize$.]{\begin{tikzpicture}
\begin{axis}[
width=0.49\textwidth,
xmin=7,xmax=17,
ymin=0,ymax=0.0051,
enlargelimits=false,
xtick={8,12,16},
xticklabels={$2^8$,$2^{12}$,$2^{16}$},
yticklabels={$0\%$,$0.25\%$,$0.50\%$},
ytick={0,0.0025,0.005},
scaled ticks=false,
xlabel={$\FieldSize$},
grid=major
]
\addplot+[mark options={fill=white,draw=blue}] table[x={logn},y={pcor24_21}]{plotn.dat}; \addlegendentry{$\VeatureSize=24$}
\addplot+[mark=x] table[x={logn},y={pcor38_35}]{plotn.dat}; \addlegendentry{$\VeatureSize=38$}
\addplot+[mark=+] table[x={logn},y={pcor44_41}]{plotn.dat}; \addlegendentry{$\VeatureSize=44$}
\end{axis}
\end{tikzpicture}}
\hspace{0.02\textwidth}
\subfigure[For $\SecretSize=\VeatureSize-1$, the average likelihood $\prob{cor}(\lceil(\VeatureSize+\SecretSize)/2-1\rceil)$ plotted versus the field size $\FieldSize$.]{\begin{tikzpicture}
\begin{axis}[
width=0.49\textwidth,
xmin=7,xmax=17,
ymin=0,ymax=0.035,
enlargelimits=false,
xtick={8,12,16},
xticklabels={$2^8$,$2^{12}$,$2^{16}$},
yticklabels={$0\%$,$1\%$,$2\%$,$3\%$},
ytick={0,0.01,0.02,0.03},
scaled ticks=false,
xlabel={$\FieldSize$},
grid=major
]
\addplot+[mark options={fill=white,draw=blue}] table[x={logn},y={pcor24_23}]{plotn.dat}; \addlegendentry{$\VeatureSize=24$}
\addplot+[mark=x] table[x={logn},y={pcor38_37}]{plotn.dat}; \addlegendentry{$\VeatureSize=38$}
\addplot+[mark=+] table[x={logn},y={pcor44_43}]{plotn.dat}; \addlegendentry{$\VeatureSize=44$}
\end{axis}
\end{tikzpicture}}
\caption{Plots for $\prob{cor}(\NumOverlap)$ at the critical $\NumOverlap=\lceil(\VeatureSize+\SecretSize)/2-1\rceil$.}
\label{fig:pcor_for_varying_field_size}
\end{center}
\ifdouble
\end{figure*}
\else
\end{figure}
\fi
\fi

As stated before, for $\NumOverlap<\lceil(\VeatureSize+\SecretSize)/2-1\rceil$ no correct outputs were observed, but nevertheless, Algorithm \ref{Algo:partial_recovery} sometimes output a (spurious) triple. Our experiments indicate that $\prob{out}(\NumOverlap)$ is independent of $\NumOverlap$ as long as $\NumOverlap< \lceil(\VeatureSize+\SecretSize)/2-1\rceil$: For fixed $\FieldSize,\VeatureSize,\SecretSize$, the variance of the observed frequencies $\prob{out}(0),\hdots,\prob{out}(\lceil(\VeatureSize+\SecretSize)/2-2\rceil)$ was close to zero, specifically, smaller than $3.6\cdot 10^{-6}$. Therefore, we considered the probability $\prob{out}(\{\NumOverlap<\lceil(\VeatureSize+\SecretSize)/2-2\rceil\})$, which is averaged over all $\NumOverlap \in [0,\lceil(\VeatureSize+\SecretSize)/2-2\rceil]$. We observed that $\prob{out}(\{\NumOverlap<\lceil(\VeatureSize+\SecretSize)/2-1\rceil\})$ decreases exponentially as $\VeatureSize-\SecretSize$ increases, and does not depend on $\VeatureSize$ (Figure \ref{fig:pout_for_varying_t}) or on the field size $\FieldSize$ (Figure \ref{fig:pout_for_varying_field_sizes}).

\ifarxiv
\renewcommand{\figax}{0.55}
\renewcommand{\figay}{0.95}
\renewcommand{\figbx}{0.55}
\renewcommand{\figby}{0.75}
\fi
\ifsingle
\renewcommand{\figax}{0.4}
\renewcommand{\figay}{0.95}
\renewcommand{\figbx}{0.35}
\renewcommand{\figby}{0.95}
\fi
\ifdouble
\renewcommand{\figbx}{0.325}
\fi
\ifdouble
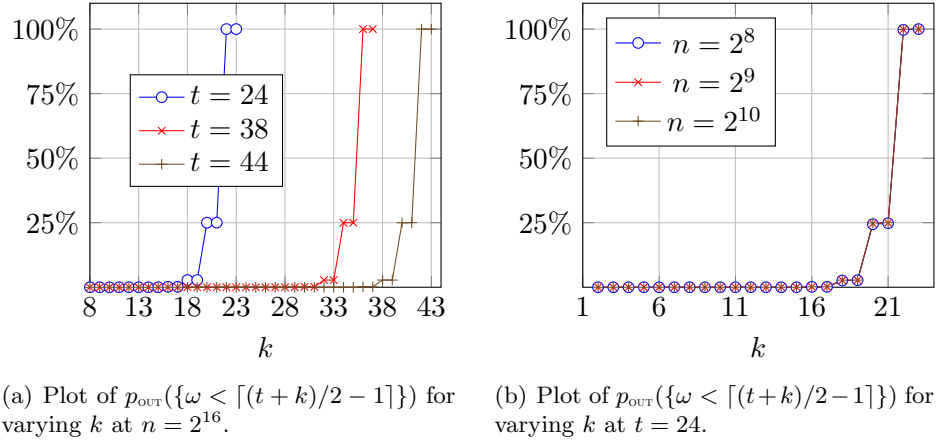
\begin{figure*}[t]
\else
\begin{figure}[!t]
\fi
\begin{center}
\subfigure[Plot of $\prob{out}(\{\NumOverlap<\lceil(\VeatureSize+\SecretSize)/2-1\rceil\})$ for varying $\SecretSize$ at $\FieldSize=2^{16}$.]{\label{fig:pout_for_varying_t}\begin{tikzpicture}
\begin{axis}[
legend style={at={(\figbx,\figby)}},
width=0.49\textwidth,
xmin=8,xmax=44,
ymin=0,ymax=1.1,
enlargelimits=false,
xtick={8,13,18,23,28,33,38,43},
yticklabels={$25\%$,$50\%$,$75\%$,$100\%$}, 
ytick={0.25,0.5,0.75,1},
scaled ticks=false,
xlabel={$\SecretSize$},
grid=major
]
\addplot+[mark options={fill=white,draw=blue}] table[x={k},y={pout16}]{plot24.dat}; \addlegendentry{$\VeatureSize=24$}
\addplot+[mark=x] table[x={k},y={pout16}]{plot38.dat}; \addlegendentry{$\VeatureSize=38$}
\addplot+[mark=+] table[x={k},y={pout16}]{plot44.dat}; \addlegendentry{$\VeatureSize=44$}
\end{axis}
\end{tikzpicture}}
\hspace{0.02\textwidth}
\subfigure[Plot of $\prob{out}(\{\NumOverlap<\lceil(\VeatureSize+\SecretSize)/2-1\rceil\})$ for varying $\SecretSize$ at $\VeatureSize=24$.]{\label{fig:pout_for_varying_field_sizes}\begin{tikzpicture}
\begin{axis}[
legend style={at={(\figax,\figay)}},
width=0.49\textwidth,
xmin=1,xmax=24,
ymin=0,ymax=1.1,
enlargelimits=false,
xtick={1,6,11,16,21},
yticklabels={$25\%$,$50\%$,$75\%$,$100\%$}, 
ytick={0.25,0.5,0.75,1},
scaled ticks=false,
xlabel={$\SecretSize$},
grid=major
]
\addplot+[mark options={fill=white,draw=blue}] table[x={k},y={pout8}]{plot24.dat}; \addlegendentry{$\FieldSize=2^8$}
\addplot+[mark=x] table[x={k},y={pout9},only marks]{plot24.dat}; \addlegendentry{$\FieldSize=2^{9}$}
\addplot+[mark=+] table[x={k},y={pout10},only marks]{plot24.dat}; \addlegendentry{$\FieldSize=2^{10}$}
\end{axis}
\end{tikzpicture}}
\caption{Plots for the average frequencies at which Algorithm \ref{Algo:partial_recovery} outputs a triple provided $\NumOverlap<\lceil(\VeatureSize+\SecretSize)/2-1\rceil$.}
\end{center}
\ifdouble
\end{figure*}
\else
\end{figure}
\fi

We also kept track of the computer time required to run the partial recovery attack. In particular, for each combination of $(\FieldSize,\VeatureSize,\SecretSize,\NumOverlap)$ that we tested, running a single partial recovery attack could be performed in less than $25ms$ in average on a single core of a $2.6$ GHz server. Moreover, in order to demonstrate the efficiency of our implementation of the attack for extreme parameters, we ran $10^5$ tests for randomly chosen vault pairs where $\FieldSize=2^{16}$, $\VeatureSize=\WeatureSize=256$, $\SecretSize=200$, and $\NumOverlap=228$ which consumed approximately $34ms$ in average on the same computer.

\subsubsection{Discussion}
Our experiments show that, unless $\SecretSize$ is close to $\VeatureSize$, our attack efficiently and reliably determines if two vault records are related, \ie{} if the protected feature sets overlap in at least $(\VeatureSize+\SecretSize)/2$ elements and, if true, recovers the feature elements in which the sets differ. In implementations of the (original) fuzzy vault for biometric template protection, typically, $\SecretSize < \VeatureSize/2$, which implies that almost all outputs of the attack are correct. For example, for the parameters $\VeatureSize=24$, $\SecretSize=9,\ldots,11$, and $\FieldSize=2^{16}$ used by a minutiae-based fuzzy vault in \cite{bib:NandakumarJainPankanti2007}, all output triples observed in our experiments were correct. The same holds true for the parameters $\VeatureSize=38$ and $\SecretSize=15$ used by \cite{bib:ClancyKiyavashLin2003} and $\VeatureSize=44$ and $\SecretSize=7,\ldots,12$ used in a minutiae-based implementation eligible for the improved fuzzy vault scheme \cite{bib:Tams2013prealign}.

We conclude that the partial recovery attack is a serious risk that must be taken into account when designing systems based on the improved fuzzy vault scheme. The recovered elements may even be used to ease full recovery of the templates. In Section \ref{sec:FullRecoveryAttack}, we develop a full recovery attack based on our partial recovery attack for which it can be even proven that it is optimal in an information theoretic sense.

\iffull
\subsubsection{Analysis of Spurious Outputs}\label{sec:PartRecsExtField}
Our experiments have shown that, for $\NumOverlap<(\VeatureSize+\SecretSize)/2$, Algorithm \ref{Algo:partial_recovery} sometimes returns incorrect outputs, in particular, if $\SecretSize$ is close to $\VeatureSize$. In order to verify that Statement b) of Theorem \ref{thm:main} cannot be strengthened, \ie{} that there are indeed outputs that do not correspond to a solution in the base field $\FiniteField$, we conducted additional experiments. In these experiments, we used the outputs $(\NumOverlap^*,\VeatureSet_0,\WeatureSet_0)$ of Algorithm \ref{Algo:partial_recovery} as a starting point for a full exhaustive search for sets $\hat{\VeatureSet}\supset \VeatureSet_0$ and $\hat{\WeatureSet}\supset \WeatureSet_0$ so that $\VaultPoly=\hat{\VaultSecret}+\CharPoly_{\hat{\VeatureSet}}$ and $\WaultPoly=\hat{\WaultSecret}+\CharPoly_{\hat{\WeatureSet}}$. 

For $\FieldSize=2^8$, $\VeatureSize=10$, $\SecretSize=3$, and $\NumOverlap=5$, we observed that in $\prob{out}(\NumOverlap)\approx 0.01034\%$ of $10^7$ tests the partial recovery attack output a triple $(\NumOverlap^*,\VeatureSet_0,\WeatureSet_0)$; furthermore, we found $\NumOverlap^*\geq(\VeatureSize+\SecretSize)/2$ for $94.39072\%$ of the output triple. However, our exhaustive search found no feature sets $\hat{\VeatureSet}\supset \VeatureSet_0$ and $\hat{\WeatureSet}\supset \WeatureSet_0$ in $\FiniteField$ protected by the vaults. Hence, we confirmed that Algorithm \ref{Algo:partial_recovery} can indeed output spurious triples that do not correspond to feature sets in $\FiniteField$ --- even if a triple with $\NumOverlap^*\geq(\VeatureSize+\SecretSize)/2$ is output.
\fi

%% file: fullrec.tex
\subsection{Full Recovery Attack}\label{sec:FullRecoveryAttack}
For $\NumOverlap \geq (\VeatureSize+\SecretSize)/2$, our partial recovery attack allows to efficiently determine from vault records of sufficiently overlapping feature sets those elements that are not in both sets. If the number $\WErrorSize=\VeatureSize-\NumOverlap$ of recovered elements of the feature set $\VeatureSet$ is at least $\SecretSize$, they can be used to ease full recovery of $\VeatureSet$. Note that for $\VeatureSize \geq 3\SecretSize + 2x$ with $x\geq 0$, we have $\WErrorSize \geq \SecretSize$ whenever $\NumOverlap \leq (\VeatureSize+\SecretSize)/2 - x$. 

In the case $\WErrorSize < \SecretSize$, although the known feature elements do not allow deterministic recovery of the complete feature sets, it can increase the success probability of full recovery attacks. In \cite{bib:BlantonAliasgari2013} it was suggested to guess the remaining feature elements, resulting in a success probability of $\tbinom{\FieldSize-\VErrorSize-\WErrorSize}{\VeatureSize-\WErrorSize}$ where $\VErrorSize$ is the number of elements recovered from $\WeatureSet$. However, due to the error correction capabilities of the scheme, it suffices to guess only $\SecretSize-\WErrorSize$ many feature elements, resulting in a success probability of $\tbinom{\NumOverlap}{\SecretSize-\WErrorSize} / \tbinom{\FieldSize}{\SecretSize-\WErrorSize}$. 

We extend this approach to smaller $\NumOverlap$. The basic idea is to guess a sufficient number of elements from $\VeatureSet\setminus\WeatureSet$ and $\WeatureSet\setminus\VeatureSet$ until the sets $\VeatureSet'$ and $\WeatureSet'$ of remaining elements satisfy the condition for the partial recovery attack with correspondingly reduced parameters. To compute the vault records $\VaultPoly'$ and $\WaultPoly'$ of the reduced sets from the original vault records, we need the following lemma. 

\begin{lemma}\label{lem:coeff_char_poly}
Let $\VeatureSet$ be a feature set of size $\VeatureSize$ and $\VaultCoeff_0, \ldots,\VaultCoeff_{\VeatureSize}$ be the coefficients of the characteristic polynomial $\CharPoly_{\VeatureSet}(\SYM)$ of $\VeatureSet$. Let $\VeatureElement\in\VeatureSet$ and $\WaultCoeff_0,\ldots,\WaultCoeff_{\VeatureSize-1}$ be the coefficients of the characteristic polynomial $\CharPoly_{\VeatureSet\setminus\{\VeatureElement\}}$. Then, for $m=\SecretSize,\ldots \VeatureSize-2$
\begin{equation}\label{eq:reduction_coeff0}
\WaultCoeff_m=\VaultCoeff_{m+1}+\VeatureElement\cdot\WaultCoeff_{m+1}
\end{equation}
In particular, the coefficients $\WaultCoeff_0, \ldots,\WaultCoeff_{\VeatureSize-1}$ are given by the equations
\begin{equation}\label{eq:reduction_coeff}
\WaultCoeff_m = \sum_{i=m+1}^{\VeatureSize} \VeatureElement^{i-m-1} \VaultCoeff_{i}.
\end{equation} 
\end{lemma}
\iffull
\BEGINPROOF
From $\CharPoly_{\VeatureSet}(\SYM) = (\SYM-\VeatureElement) \CharPoly_{\VeatureSet'}(\SYM)$, we get Equation (\ref{eq:reduction_coeff0}) for $m=\SecretSize,\ldots \VeatureSize - 2$. By recursively applying this equation and using $\WaultCoeff_{\VeatureSize - 1}= \VaultCoeff_{\VeatureSize}=1$ we obtain Equation (\ref{eq:reduction_coeff}).  
\ENDPROOF
\fi

\subsubsection{The Attack}	
	
We next describe the algorithmic of our full recovery attack.

\begin{alg}[Full Recovery Attack]
\label{Algo:full_recovery}
\AlgInput Two vault records $\VaultPoly(\SYM)=\sum_{j=0}^{\VeatureSize}\VaultCoeff_j\cdot\SYM^j$ and $\WaultPoly(\SYM)=\sum_{j=0}^{\WeatureSize}\WaultCoeff_j\cdot\SYM^j$ of sets $\VeatureSet$ and $\WeatureSet$, respectively, with $\VeatureSize\geq\WeatureSize$, and a natural number $\NumOverlap'$ as candidate for $\NumOverlap=|\VeatureSet\cap\WeatureSet|$.
\AlgOutput Either a candidate pair for $(\VeatureSet,\WeatureSet)$, or \AlgFailure. 
\end{alg}

\begin{enumerate}

\item Initialize $\ThreshDist\leftarrow\max\left(0,\lceil(\VeatureSize+\SecretSize)/2\rceil - \NumOverlap'\right)$ and $\VeatureSet^*\leftarrow\emptyset$. \label{full_initialize}

\item If $\ThreshDist>0$, reduce the vault records $\VaultPoly$ and $\WaultPoly$ as follows.\label{full_1}

\begin{enumerate}

\item Guess $\ThreshDist$ elements $\VeatureElement_1,\ldots, \VeatureElement_{\ThreshDist}$ from $\VeatureSet\setminus\WeatureSet$ and use (\ref{eq:reduction_coeff}) to iteratively compute the coefficients $(\bar{\VaultCoeff}_{\SecretSize-\ThreshDist}, \ldots, \bar{\VaultCoeff}_{\VeatureSize-\ThreshDist})$ of the characteristic polynomial $\CharPoly_{\bar{\VeatureSet}}$ of $\bar{\VeatureSet} = \VeatureSet \setminus \{\VeatureElement_1,\ldots,\VeatureElement_{\ThreshDist}\}$ from $(\VaultCoeff_{\SecretSize},\ldots,\VaultCoeff_{\VeatureSize})$
and add $\{\VeatureElement_1,\ldots,\VeatureElement_{\ThreshDist}\}$ to $\VeatureSet^*$.\label{full_fill1}

\item Guess $\ThreshDist$ elements $\WeatureElement_1,\ldots, \WeatureElement_{\ThreshDist}$ from $\WeatureSet\setminus\VeatureSet$ and use (\ref{eq:reduction_coeff}) to iteratively compute the coefficients $(\bar{\WaultCoeff}_{\SecretSize-\ThreshDist}, \ldots, \bar{\WaultCoeff}_{\WeatureSize-\ThreshDist})$ of the characteristic polynomial $\CharPoly_{\bar{\WeatureSet}}$ of $\bar{\WeatureSet} = \WeatureSet \setminus \{\WeatureElement_1,\ldots,\WeatureElement_{\ThreshDist}\}$ from $(\WaultCoeff_{\SecretSize},\ldots,\WaultCoeff_{\WeatureSize})$. 
\label{full_fill2}

\end{enumerate}

Otherwise, if $\ThreshDist=0$, let $(\bar{\VaultCoeff}_\SecretSize,\ldots,\bar{\VaultCoeff}_\VeatureSize)\leftarrow(\VaultCoeff_\SecretSize,\ldots,\VaultCoeff_\VeatureSize)$ and $(\bar{\WaultCoeff}_\SecretSize,\ldots,\bar{\WaultCoeff}_\WeatureSize)\leftarrow(\WaultCoeff_\SecretSize,\ldots,\VaultCoeff_\WeatureSize)$.

\item Invoke the partial recovery attack (Algorithm~\ref{Algo:partial_recovery}) with input 
$\sum_{j=\bar{\SecretSize}}^{\bar{\VeatureSize}}\bar{\VaultCoeff}_j\cdot\SYM^j$ and $\sum_{j=\bar{\SecretSize}}^{\bar{\WeatureSize}}\bar{\WaultCoeff}_j\cdot\SYM^j$ where $\bar{\VeatureSize}=\VeatureSize-\ThreshDist$, $\bar{\WeatureSize}=\WeatureSize-\ThreshDist$ and $\bar{\SecretSize}=\SecretSize-\ThreshDist$ (here, $\bar{\SecretSize}$ is the degree of the secret polynomials); if the partial recovery attack returns \AlgFailure, do so as well;
otherwise, if $(\NumOverlap^*,\VeatureSet_0,\WeatureSet_0)$ is the output of Algorithm~\ref{Algo:partial_recovery}, add $\VeatureSet^*\leftarrow\VeatureSet_0\cup\VeatureSet^*$ to $\VeatureSet^*$ and update $\NumOverlap'\leftarrow\NumOverlap^*$. \label{algo_ful:step_partial}

\item If $\WErrorSize'=\VeatureSize-\NumOverlap' < \SecretSize$, guess $\LastGuess=\SecretSize - \WErrorSize'$ many elements $\VeatureElement_{\WErrorSize'+1},\hdots,\VeatureElement_{\SecretSize}$ from $\VeatureSet \cap \WeatureSet$ and add them to $\VeatureSet^* $. \label{full_3}\label{full_complement}

\item Unlock $\VaultPoly$ using $\VeatureSet^*$: Compute the unique polynomial $\VaultSecret^*$ of degree smaller than $\SecretSize$ interpolating the pairs $(\VeatureElement,\VaultPoly(\VeatureElement))$ for all $\VeatureElement\in \VeatureSet^*$. If $\VaultPoly-\VaultSecret^*$ splits into distinct linear factors, set $\VeatureSet'$ as the set of its roots; otherwise return \AlgFailure.\label{full_unlock}

\item Set $\OverlapSet=\VeatureSet' \setminus (\VeatureSet_0 \cup \{\VeatureElement_1,\ldots,\VeatureElement_{\ThreshDist}\})$ as candidate for $\VeatureSet \cap \WeatureSet$, set $\WeatureSet'=\WeatureSet_0 \cup \{\WeatureElement_1,\ldots,\WeatureElement_{\ThreshDist}\} \cup \OverlapSet$, and output $(\VeatureSet', \WeatureSet')$.\label{full_last}

\end{enumerate}

If $\NumOverlap\geq(\VeatureSize+\SecretSize)/2$, the full recovery attack should be run with input $\NumOverlap'=\lceil(\VeatureSize+\SecretSize)/2\rceil$ which will be updated to the correct $\NumOverlap$ in Step \ref{algo_ful:step_partial}; otherwise, if $\NumOverlap<(\VeatureSize+\SecretSize)$, the full recovery attack requires that $\NumOverlap'=\NumOverlap$ has been guessed correctly; even if $\NumOverlap'=\NumOverlap$, the attack may return \AlgFailure~and this happens if the elements guessed in Step \ref{full_fill1}, \ref{full_fill2}, or \ref{full_complement} are incorrect of which probability $\prob{Full}$ can be lower bounded by the following theorem.

\subsubsection{Analysis}

\begin{thm}\label{thm:full}
Assume $\VeatureSize\geq\WeatureSize\geq \SecretSize$ and $\NumOverlap'=\NumOverlap$. Then Algorithm~\ref{Algo:full_recovery} outputs $(\VeatureSet, \WeatureSet)$ with probability 
\begin{equation*} 
\prob{Full} \geq 
\frac{\binom{\VeatureSize-\NumOverlap}{\ThreshDist}\binom{\WeatureSize-\NumOverlap}{\ThreshDist}\binom{\NumOverlap}{\LastGuess}}{\binom{\FieldSize}{\ThreshDist}^2 \binom{\FieldSize}{\LastGuess}}  
\end{equation*}
where $\ThreshDist= \max\left(0,\lceil(\VeatureSize+\SecretSize)/2\rceil - \NumOverlap\right)$ and $\LastGuess=\max\left(0,\SecretSize - \WErrorSize\right)$.

Furthermore, if $\NumOverlap > \max(\SecretSize,\VeatureSize/2)+\ParityBit(\VeatureSize+\SecretSize)$, where $\ParityBit(\SomeInteger)$ is the parity bit of an integer $\SomeInteger$, the number of elements that Algorithm~\ref{Algo:full_recovery} guesses in Step \ref{full_1} and \ref{full_3} is smaller than $\SecretSize$, and consequently, $\prob{Full} \in \Oh(\FieldSize^{\SecretSize-1})$ for fixed $\VeatureSize$.
\end{thm}

\iffull
\BEGINPROOF
With probability 
\begin{equation*} 
\frac{\binom{\VeatureSize-\NumOverlap}{\ThreshDist}\binom{\WeatureSize-\NumOverlap}{\ThreshDist}}{\binom{\FieldSize}{\ThreshDist}^2}  
\end{equation*}
the guesses in Step \ref{full_1} are correct. In this case, $\NumOverlap\geq (\bar{\VeatureSize}+\bar{\SecretSize})/2$ and, hence, the partial recovery attack outputs 
$\VeatureSet_0  = \bar{\VeatureSet}\setminus \bar{\WeatureSet} = (\VeatureSet\setminus \WeatureSet)\setminus \{\VeatureElement_1,\ldots,\VeatureElement_{\ThreshDist}\}$
and $\WeatureSet_0  = (\WeatureSet\setminus \VeatureSet)\setminus \{\WeatureElement_1,\ldots,\WeatureElement_{\ThreshDist}\}$. 

Furthermore, the probability that the guesses in Step \ref{full_3} are correct is at least $\tbinom{\NumOverlap}{\LastGuess}/ \tbinom{\FieldSize}{\LastGuess}$. Assuming success in the previous steps, we have $\VeatureSet^*=\VeatureSet\setminus \WeatureSet$, which has $\WErrorSize$ elements. If, in Step \ref{full_3}, all $\SecretSize-\WErrorSize$ elements were correctly guessed from $\VeatureSet\cap \WeatureSet$, then $\VeatureSet^*$ contains $\SecretSize$ elements from $\VeatureSet$, and thus, we obtain $\VeatureSet'=\VeatureSet$. The probability that these guesses are correct is at least $\tbinom{\NumOverlap}{\LastGuess} / \tbinom{\FieldSize}{\LastGuess}$. 

Finally, provided that the guesses in Step \ref{full_1} were correct, we get $\OverlapSet = \VeatureSet\cap \WeatureSet$ and, therefore, $\WeatureSet'=\WeatureSet$. Multiplying all probabilities completes the proof of the first claim.

The number of elements guessed in Step \ref{full_fill1}, \ref{full_fill2} and \ref{full_3} is $2\ThreshDist+m$. Since both variables, $\ThreshDist$ and $\LastGuess$ are defined with a maximum function, we need to distinguish different cases. 

First consider the case $\WErrorSize < \SecretSize$, which implies $\LastGuess=\SecretSize - \WErrorSize$. In this case, if $\NumOverlap < \lceil(\VeatureSize+\SecretSize)/2\rceil$, we have $\ThreshDist = \lceil(\VeatureSize+\SecretSize)/2\rceil - \NumOverlap$, and, hence, $2\ThreshDist+\LastGuess$ evaluates to $2\SecretSize-\NumOverlap$ if $\VeatureSize+\SecretSize$ is even, and to $2\SecretSize-\NumOverlap+1$ if $\VeatureSize+\SecretSize$ is odd. Since $\NumOverlap > \SecretSize+\ParityBit(\VeatureSize+\SecretSize)$, in either case, the result is smaller than $\SecretSize$. On the other hand, if  $\NumOverlap \geq \lceil(\VeatureSize+\SecretSize)/2\rceil$, we have $\ThreshDist = 0$, and  $2\ThreshDist+\LastGuess$ evaluates to $\SecretSize - \WErrorSize$ which is obviously smaller than $\SecretSize$.

Now consider the case $\WErrorSize \geq \SecretSize$, which implies $\LastGuess=0$. In this case, if $\NumOverlap < \lceil(\VeatureSize+\SecretSize)/2\rceil$, we have $\ThreshDist = \lceil(\VeatureSize+\SecretSize)/2\rceil - \NumOverlap$, and, hence, $2\ThreshDist+\LastGuess$ evaluates to $\VeatureSize+\SecretSize-2\NumOverlap$ if $\VeatureSize+\SecretSize$ is even, and to $\VeatureSize+\SecretSize-2\NumOverlap+1$ if $\VeatureSize+\SecretSize$ is odd. Since $\NumOverlap > \VeatureSize/2+\ParityBit(\VeatureSize+\SecretSize)$, in either case, the result is smaller than $\SecretSize$. On the other hand, if  $\NumOverlap \geq \lceil(\VeatureSize+\SecretSize)/2\rceil$, we have $\ThreshDist = 0$, and $2\ThreshDist+\LastGuess$ evaluates to zero. 
\ENDPROOF
\fi

Theorem \ref{thm:full} shows that for  $\NumOverlap > \max(\SecretSize,\VeatureSize/2)+\ParityBit(\VeatureSize+\SecretSize)$ and large finite fields, our full recovery attack is
better than the obvious attack that guesses $\SecretSize$ elements of $\VeatureSet \cap \WeatureSet$ and uses them to recover both feature sets. For smaller $\NumOverlap$, however, it is more efficient  to guess $\min(\SecretSize,\NumOverlap)$ elements of $\VeatureSet \cap \WeatureSet$ and, if $\NumOverlap<\SecretSize$, additional $\SecretSize-\NumOverlap$ elements of each of the difference sets $\VeatureSet\setminus\WeatureSet$ and $\WeatureSet\setminus\VeatureSet$. 
\iffull
In Section \ref{sec:bounds}, we will show that our attack is, for $\NumOverlap \geq \min\left(\VeatureSize-\SecretSize, (\VeatureSize+\SecretSize)/2\right)$, even asymptotically optimal in an information theoretical sense, when considering its probability as a function in $\FieldSize$. 
\fi

\iffull
\subsubsection{Brute-Force Attack}\label{sec:FullRecoveryDiscussion}
While executing Algorithm~\ref{Algo:full_recovery} reveals the feature sets with the probability estimated in Theorem \ref{thm:full}, a real attacker typically wishes to increase the success probability at the cost of computation time. This can obviously be achieved by repeating Algorithm~\ref{Algo:full_recovery} with different guesses, if it has been unsuccessful. 

The attacker knows that the full recovery has failed as soon as the partial recovery attack invoked in Step \ref{algo_ful:step_partial} returns \AlgFailure~or when unlocking $\VaultPoly$ fails in Step \ref{full_unlock}. Furthermore, in some circumstances, additional information may be available to distinguish the correct from wrong outputs; for example, many implementations of the original fuzzy vault encode a CRC code into the secret polynomials \cite{bib:UludagPankantiJain2005,bib:UludagJain2006,bib:NandakumarJainPankanti2007,bib:NandakumarNagarJain2007,bib:Nagar2010} or store a cryptographic hash value of its coefficients \cite{bib:MerkleEtAl2011MultiFingerVault,bib:Tams2013prealign} to aid verification of its correctness. In such a case, most wrong outputs can be detected using an additional final check and the attack can be retried. Otherwise, \ie{} if such a check is not possible, the attacker can only output a list of all candidate feature sets, the size of which depends on the parameters (see \cite{bib:JuelsSudan2006}).

Clearly, it is not necessary to repeat the complete Algorithm~\ref{Algo:full_recovery}. Instead the following iteration approach can be applied. For each successful output of the partial recovery attack in Step \ref{algo_ful:step_partial}, the attacker can fix the output and successively try all tuples $(\VeatureElement_{\WErrorSize'+1},\hdots,\VeatureElement_{\SecretSize})$ in Step \ref{full_complement}. For fixed $\NumOverlap'\geq (\VeatureSize + \SecretSize)/2$, this already represents a complete exhaustive search attack. On the other hand, if $\NumOverlap'< (\VeatureSize + \SecretSize)/2$, the outcome in Step \ref{algo_ful:step_partial} depends on the guesses in Step \ref{full_1} and thus an attacker facing an error in Step \ref{algo_ful:step_partial} (if  Algorithm~\ref{Algo:partial_recovery} returns \AlgFailure), in Step \ref{full_unlock} (if unlocking $\VaultPoly$ fails), or in the additional final check for correctness (if implemented), can go back to Step \ref{full_initialize} and repeat the complete Algorithm~\ref{Algo:full_recovery} choosing in Step \ref{full_1} the next $\VeatureElement_1,\ldots, \VeatureElement_{\ThreshDist}$ and $\WeatureElement_1,\ldots, \WeatureElement_{\ThreshDist}$. 

If $\NumOverlap$ is unknown, the attack should start with $\NumOverlap'= \lceil (\VeatureSize + \SecretSize)/2 \rceil$ for two reasons. First, for this input, no guessing is done in Step \ref{full_1} resulting in an exponentially smaller search space as compared to smaller values of $\NumOverlap'$ (where we have $h >0$); second, if $\NumOverlap\geq \NumOverlap' = \lceil (\VeatureSize + \SecretSize)/2 \rceil$,  Algorithm~\ref{Algo:partial_recovery} is guaranteed to succeed in Step \ref{algo_ful:step_partial}. If trying $\NumOverlap'$ does not yield the correct result, the attack can proceed by successively decreasing $\NumOverlap'$, each decrementation resulting in an increase of the search space by a factor of $\FieldSize^2$. When the search space becomes too large, \ie{} at a certain value of $\NumOverlap'$, the attack has to abort. 

An exhaustive search algorithm can be obtained by replacing the random guesses in Algorithm~\ref{Algo:full_recovery} by systematically trying all tuples. However, a probabilistic method is more efficient, because, typically, there exist many tuples that will result in a correct output. If $\FullRecFraction$ is the fraction of tuples in the search space resulting in a correct output, a probabilistic search will, on average, succeed after $1 / \FullRecFraction$ trials and, according to \emph{Markov's inequality}, will succeed with probability at least $1-1/\FullRecConfidence$, for any $\FullRecConfidence > 1$, after $\FullRecConfidence / \FullRecFraction$ trials. For instance, in Step \ref{full_complement}, provided that the previous steps have been successful, $\tbinom{\VeatureSize-\NumOverlap}{\SecretSize-\WErrorSize}$ out of all possible tuples will result in a correct output. Thus, for $\FullRecProportion=\tbinom{\VeatureSize-\NumOverlap}{\SecretSize-\WErrorSize}/\FullRecConfidence>1$, trying only a random $\FullRecProportion$-fraction of all possible tuples, the attacker will gain a speed-up of $\FullRecProportion$ as compared to exhaustive search, while having an error error probability of at most $1/\FullRecConfidence$. 

On the other hand, an exhaustive search algorithm can be useful for \analyzing{} purposes. In Section \ref{sec:PartRecsExtField}, for example, we used exhaustive search to experimentally determine the frequency with which the outputs of the partial recovery attack corresponds to a solutions in $\FiniteField$.
\fi

%% file: bounds.tex
\section{Theoretical Bounds for Recovery Attacks}\label{sec:bounds}
In this section, we derive bounds for the success probability of partial and full recovery attacks. These bounds are based on an upper bound for the information leakage, \ie{} a lower bound for the average min-entropy, of two vault records based on overlapping feature sets. Our bounds are based on information theoretical results and, thus, hold even for attackers with unlimited computational resources, \ie{} also for attackers with exponential running time. Of course, these results assume that the attacker is not able to verify the correctness of a candidate feature set by additional information, \eg{} by a stored hash value or a CRC code of the secret polynomial $\VaultSecret$, as often suggested in implementations of the original fuzzy vault \cite{bib:UludagPankantiJain2005,bib:UludagJain2006,bib:NandakumarJainPankanti2007,bib:MerkleEtAl2011MultiFingerVault,bib:Tams2013prealign}.

\subsection{Bounds for Full Recovery Attacks}\label{bound_full}
The following result bounds the amount of information leaked from two vault records computed from overlapping feature sets.
\iffull
The proof of the theorem is given in the appendix.
\fi
\begin{thm}\label{EntropyLoss}
Let $\VeatureSet$ and $\WeatureSet$ be two feature sets of size $\VeatureSize$ and $\WeatureSize\leq \VeatureSize$, respectively. Then the entropy loss of vault records $\VaultPoly$ and $\WaultPoly$ computed from $\VeatureSet$ and $\WeatureSet$, respectively, is at most $\min(\VeatureSize+\WeatureSize-2\SecretSize,\VeatureSize-\SecretSize + \diffVW) \log \FieldSize$,
\iffull
\ie{}
\ifdouble
\begin{align}
\begin{split}
\mathbf{\tilde{H}_{\infty}}(\VeatureSet,\WeatureSet | \VaultPoly,\WaultPoly) \geq &\mathbf{H_{\infty}}(\VeatureSet,\WeatureSet) \nonumber \\
&\; - \min\left(\VeatureSize+\WeatureSize-2\SecretSize,\VeatureSize-\SecretSize + \diffVW\right) \log \FieldSize, \nonumber 
\end{split}
\end{align}
\else
\begin{equation*}
\mathbf{\tilde{H}_{\infty}}(\VeatureSet,\WeatureSet | \VaultPoly,\WaultPoly) \geq \mathbf{H_{\infty}}(\VeatureSet,\WeatureSet) - \min\left(\VeatureSize+\WeatureSize-2\SecretSize,\VeatureSize-\SecretSize + \diffVW\right) \log \FieldSize,
\end{equation*}
\fi
\fi
where $\diffVW=|(\VeatureSet \cup \WeatureSet)\setminus (\VeatureSet \cap \WeatureSet)|$ is the set difference between $\VeatureSet$ and $\WeatureSet$.
\end{thm}

If the feature sets overlap in $\NumOverlap$ elements, we have $\diffVW=\VeatureSize+\WeatureSize-2\NumOverlap$. Thus, for $\NumOverlap\geq (\VeatureSize+\SecretSize)/2$, Theorem \ref{EntropyLoss} gives an estimate for the information leakage of $(2\VeatureSize+\WeatureSize-2\NumOverlap-\SecretSize)\log\FieldSize$, whereas, for $\NumOverlap\leq (\VeatureSize+\SecretSize)/2$, we obtain a bound of $(\VeatureSize+\WeatureSize-2\SecretSize)\log\FieldSize$. 

Since the lower $\SecretSize$ coefficients of the vault record are either zero (deterministic version) or independent of the features and the remaining coefficients are deterministic, the estimation provided by Theorem \ref{EntropyLoss} does not only hold in the average case for random vault records, but also for any fixed vault records $\VaultPoly,\WaultPoly$. Thus, we obtain the following corollary.

\begin{cor}\label{BoundFull}
Any algorithm that takes as input two  vault records $\VaultPoly$ and $\WaultPoly$ computed from feature sets $\VeatureSet$ and $\WeatureSet$ of size $\VeatureSize$ and $\WeatureSize\leq \VeatureSize$, respectively, overlapping in $\NumOverlap$ elements, and outputs $\VeatureSet, \WeatureSet$ has success probability  
\begin{equation}
\prob{Full} \leq 2^{-\mathbf{H_{\infty}}(\VeatureSet,\WeatureSet)}\cdot \FieldSize^{\min\left(\VeatureSize+\WeatureSize-2\SecretSize,\, 2\VeatureSize + \WeatureSize - 2\NumOverlap-\SecretSize\right)}.
\end{equation}
\end{cor}

Observe that for $\VeatureSize \geq 3\SecretSize$, Theorem \ref{EntropyLoss} does not give positive bounds on the average min-entropy for any $\NumOverlap \in [2 \SecretSize,\VeatureSize - \SecretSize]$, and thus, Corollary \ref{BoundFull} does not give a meaningful upper bound for attacks. 
\iffull
In this case, $(\VeatureSize+\SecretSize)/2 \leq \VeatureSize -\SecretSize$. 
\fi
While for $\NumOverlap \in [\lceil(\VeatureSize+\SecretSize)/2\rceil,\VeatureSize - \SecretSize ]$, the vaults records can indeed be uncovered deterministically, \ie{} without any guessing, using Algorithm \ref{Algo:full_recovery}, it is not clear if the vault records become insecure for $\NumOverlap \in [2 \SecretSize,\lfloor(\VeatureSize+\SecretSize)/2\rfloor-1]$. 

We now \analyze{} to which extent and for which parameters the success probability of Algorithm \ref{Algo:full_recovery} is optimal. Obviously, this is the case, when it becomes deterministic, \ie{} for $\NumOverlap \in [\lceil(\VeatureSize+\SecretSize)/2\rceil,\VeatureSize - \SecretSize ]$. However, as soon as feature elements need to be guessed, no attack can have optimal success probability for arbitrary feature spaces, unless it takes into account the specific statistical distribution. Therefore, we subsequently focus on the ideal situation, in which the feature sets $\VeatureSet$ and $\WeatureSet$ have maximal entropy. For this case, the following corollary pins down the parameters, for which our full recovery attack's success probability is asymptotically optimal in $\FieldSize$ up to a constant. 

\begin{cor}\label{cor:optimal_full}
Assume that the elements in the feature sets are independently and uniformly chosen at random, so that the feature sets $\VeatureSet$ and $\WeatureSet$ overlap in $\NumOverlap$ elements, where $\NumOverlap \geq \VeatureSize-\SecretSize$ with $\VeatureSize+\SecretSize$ being even, or $\NumOverlap \geq (\VeatureSize+\SecretSize)/2$. Then the success probability of Algorithm \ref{Algo:full_recovery} is optimal for $\FieldSize \rightarrow \infty$ up to a constant. 
\end{cor}

\iffull
\begin{proof}
For independently and uniformly chosen feature elements, we have 
\begin{equation*}
2^{\mathbf{H_{\infty}}(\VeatureSet,\WeatureSet)} =  \binom{\FieldSize}{\NumOverlap} \binom{\FieldSize-\NumOverlap}{\VeatureSize-\NumOverlap} \binom{\FieldSize-\VeatureSize}{\WeatureSize - \NumOverlap},
\end{equation*}
which is in $\Oh(\FieldSize^{\VeatureSize+\WeatureSize-\NumOverlap})$ for $\FieldSize \rightarrow \infty$.

Theorem \ref{EntropyLoss} gives different bounds for $\NumOverlap\geq (\VeatureSize+\SecretSize)/2$ and for $\NumOverlap< (\VeatureSize+\SecretSize)/2$. Consequently, we will distinguish these two main cases in our analysis. 

We first consider the case $\NumOverlap\geq (\VeatureSize+\SecretSize)/2$. In this case, Corollary \ref{BoundFull} gives an upper bound of $\prob{Full}  \leq  \FieldSize^{2\VeatureSize + \WeatureSize - 2\NumOverlap-\SecretSize}/2^{\mathbf{H_{\infty}}(\VeatureSet,\WeatureSet)}$ for the success probability $\prob{Full}$ of attackers trying to determine the feature sets from two vault records. This implies $\prob{Full} \in \Oh\left(\FieldSize^{-(\SecretSize - \WErrorSize) }\right)$, where, as before, $\WErrorSize = \VeatureSize - \NumOverlap$.

In comparison, for $\NumOverlap\geq (\VeatureSize+\SecretSize)/2$ (which implies $\ThreshDist= 0$), Theorem \ref{thm:full} estimates the success probability $\prob{Alg\ref{Algo:full_recovery}}$ of Algorithm \ref{Algo:full_recovery} as $\prob{Alg\ref{Algo:full_recovery}}\geq\tbinom{\NumOverlap}{\LastGuess} / \tbinom{\FieldSize}{\LastGuess}$. If $\NumOverlap \leq \VeatureSize-\SecretSize$, we have $\LastGuess=0$ and our full recovery attack is deterministic, which is obviously optimal. If, on the other hand, $\NumOverlap > \VeatureSize-\SecretSize$, the estimate of Theorem \ref{thm:full} evaluates to $\prob{Alg\ref{Algo:full_recovery}}\geq \tbinom{\NumOverlap}{\SecretSize-\WErrorSize} / \tbinom{\FieldSize}{\SecretSize-\WErrorSize}$, which implies $\prob{Alg\ref{Algo:full_recovery}}\in\Omega\left(\FieldSize^{-(\SecretSize - \WErrorSize) }\right)$.

Now consider $\NumOverlap< (\VeatureSize+\SecretSize)/2$. In this case, Corollary \ref{BoundFull} gives the upper bound 
$\prob{Full}  \leq  \FieldSize^{\VeatureSize+\WeatureSize-2\SecretSize} / 2^{\mathbf{H_{\infty}}(\VeatureSet,\WeatureSet)}$. 
This yields $\prob{Full}  \in \Oh\left( \FieldSize^{- (2\SecretSize-\NumOverlap)}\right)$.

On the other hand, for $\NumOverlap< (\VeatureSize+\SecretSize)/2$, we have $\ThreshDist= \lceil(\VeatureSize+\SecretSize)/2\rceil - \NumOverlap$. Thus, Theorem \ref{thm:full} gives an estimate of $\prob{Alg\ref{Algo:full_recovery}} \geq \binom{\VeatureSize-\NumOverlap}{\ThreshDist}^2 \binom{\NumOverlap}{\LastGuess} \binom{\FieldSize}{\ThreshDist}^{-2} \binom{\FieldSize}{\LastGuess}^{-1}$, which gives $\prob{Alg\ref{Algo:full_recovery}} \in \Omega\left(\FieldSize^{-(2 \ThreshDist + \LastGuess) }\right)$. By assumption, for $\NumOverlap< (\VeatureSize+\SecretSize)/2$, we have $\NumOverlap \geq \VeatureSize-\SecretSize$ and $\VeatureSize+\SecretSize$ being even, which implies 
$\LastGuess= \SecretSize+\NumOverlap-\VeatureSize$ and $2\ThreshDist=\VeatureSize+\SecretSize- 2\NumOverlap$. This gives the result for the case $\NumOverlap< (\VeatureSize+\SecretSize)/2$. 
\end{proof}
\fi

For $\NumOverlap \in [\lceil(\VeatureSize+\SecretSize)/2\rceil,\VeatureSize - \SecretSize ]$ with odd $\VeatureSize+\SecretSize$, the success probability of Algorithm \ref{Algo:full_recovery} is less optimal due to rounding of the value $\ThreshDist$. 
\iffull
However, solving the equations of Blanton and Aliasgari described in Section \ref{sec:PartialRecoveryAttack} after guessing $\VeatureSize+\SecretSize-2\NumOverlap$ elements from $(\VeatureSet\cup\WeatureSet)\setminus (\VeatureSet\cap\WeatureSet)$ gives an attack with optimal success probability even for odd $\VeatureSize+\SecretSize$, albeit being computationally inefficient for large $\VeatureSize$.

For $\NumOverlap < \min\left(\VeatureSize-\SecretSize, (\VeatureSize+\SecretSize)/2\right)$, Algorithm \ref{Algo:full_recovery} has success probability $\prob{} \in \Omega\left(\FieldSize^{-(\VeatureSize+\SecretSize)}\right)$ which is not asymptotically optimal. However, for $\NumOverlap \leq \SecretSize$, the obvious attack that guesses all $\NumOverlap$ elements from $\VeatureSet \cap \WeatureSet$ and $\SecretSize-\NumOverlap$ elements from each $\VeatureSet \setminus\WeatureSet$ and $\WeatureSet \setminus\VeatureSet$ has, for $\VeatureSize=\WeatureSize$,  success probability  
$\prob{Full} \geq \tbinom{\VeatureSize-\NumOverlap}{\SecretSize-\NumOverlap}^2 / \tbinom{\FieldSize}{\SecretSize-\NumOverlap}^{-2}\tbinom{\FieldSize}{\NumOverlap}$, which is tight with the asymptotic upper bound $\Oh\left( \FieldSize^{- (2\SecretSize-\NumOverlap)}\right)$ derived from Corollary \ref{BoundFull}. Thus, this simple attack is asymptotically optimal for $\NumOverlap \leq \SecretSize$.  
\fi

\subsection{Lower Bounds for Partial Recovery Attacks}\label{sec:bounds_partial}
As pointed out in \cite{bib:DodisSmith05}, bounds on the average min-entropy do not, in general, guarantee that no partial information can be extracted. Our partial recovery attack is an example, and, in fact, the amount of information we are able to extract ($\diffVW$ feature elements) is still smaller than the bound on the information leakage given in Theorem \ref{EntropyLoss}, which corresponds to at least $\VeatureSize -\SecretSize + \diffVW$ elements. Thus, the question arises if it was possible to recover even more feature elements from vault records of overlapping feature sets. In this section, for a wide range of parameters, we give a negative answers to this question when considering field-independent attacks and, hence, show that our partial recovery attack is indeed optimal with respect to the number of extracted elements. 

The following result shows that, for $\NumOverlap > \VeatureSize -\SecretSize$, no algorithm can determine elements from the intersection of two feature sets from corresponding vault records with probability asymptotically better than $\Oh(\FieldSize^{-1})$.

\begin{thm}\label{Thm_partial_1} 
For fixed $\VeatureSize\geq \WeatureSize \geq \NumOverlap > \VeatureSize -\SecretSize$, any algorithm that takes as input two vault records of uniformly chosen feature sets $\VeatureSet$ and $\WeatureSet$ of size $\VeatureSize$ and $\WeatureSize$, respectively, with $|\VeatureSet \cap \WeatureSet|=\NumOverlap$, and outputs an element $\VeatureElement \in \VeatureSet \cap \WeatureSet$, has success probability $\prob{}\in \Oh(\FieldSize^{-1})$. Specifically, 
\begin{equation*}
\prob{} \, \leq \,\FieldSize^{\VeatureSize-\SecretSize}\, \frac{\binom{\FieldSize}{\SecretSize-1}} 
{\binom{\FieldSize}{\VeatureSize}\binom{\VeatureSize-1}{\SecretSize -1}}.
\end{equation*}
\end{thm}

\iffull
\BEGINPROOF
Assume that $\AlgA$ is an algorithm that takes as input two vault records $\VaultPoly$ and $\WaultPoly$ of feature sets $\VeatureSet$ and $\WeatureSet$ of size $\VeatureSize$ and $\WeatureSize\leq \VeatureSize$, respectively, overlapping in $\NumOverlap > \VeatureSize -\SecretSize$ elements, and outputs an element $\VeatureElement \in \VeatureSet \cap \WeatureSet$ with probability $\prob{\AlgA}$. We will use $\AlgA$ to construct an algorithm $\AlgB$, that takes as input a single vault record $\VaultPoly$ of a feature set $\VeatureSet$  of size $\VeatureSize$ and outputs $\VeatureSet$. Then, using the bound on the average min-entropy of single vault records from \cite{bib:DodisEtAl2004}, we will derive a bound for $\prob{\AlgA}$.

Algorithm $\AlgB$ works as follows.
\begin{enumerate}

\item Guess $\WErrorSize=\VeatureSize - \NumOverlap$ many elements $\VeatureElement_1,\ldots,\VeatureElement_{\WErrorSize}$ from $\VeatureSet$ and use (\ref{eq:reduction_coeff}) to iteratively compute from the coefficients $\VaultCoeff_{\SecretSize},\ldots,\VaultCoeff_{\VeatureSize}$ of $\VaultPoly$ the coefficients $(\ZaultCoeff_{\SecretSize-\WErrorSize}, \ldots, \ZaultCoeff_{\NumOverlap})$ of the characteristic polynomial $\CharPoly_{\OverlapSet}$ of $\OverlapSet = \VeatureSet \setminus \{\VeatureElement_1,\ldots,\VeatureElement_{\WErrorSize}\}$.

\item Select $\VErrorSize=\WeatureSize - \NumOverlap$ many random elements $\{\WeatureElement_1,\ldots,\WeatureElement_{\VErrorSize}\}$ from $\FiniteField \setminus \OverlapSet$, and use Equation (\ref{eq:reduction_coeff0}) in Lemma \ref{lem:coeff_char_poly} to iteratively compute from the coefficients $(\ZaultCoeff_{\SecretSize-\WErrorSize}, \ldots, \ZaultCoeff_{\NumOverlap})$ the coefficients $\WaultCoeff_{\SecretSize},\ldots,\WaultCoeff_{\WeatureSize}$ of the characteristic polynomial $\CharPoly_{\WeatureSet}$ of $\WeatureSet = \OverlapSet  \cup \{\WeatureElement_1,\ldots,\WeatureElement_{\VErrorSize}\}$.

\item Invoke Algorithm $\AlgA$ with input $\VaultPoly$ and $\WaultPoly=(\WaultCoeff_{\SecretSize},\ldots,\WaultCoeff_{\WeatureSize})$. Let $\VeatureElement_0$ be the output of $\AlgA$.

\item Guess $\SecretSize - \WErrorSize -1$ many elements $\VeatureElement_{\WErrorSize+1},\ldots,\VeatureElement_{\SecretSize-1}$ of $\VeatureSet\setminus \{\VeatureElement_0,\VeatureElement_1,\ldots,\VeatureElement_{\WErrorSize}\}$.

\item Unlock $\VeatureSet$: compute the unique polynomial $\VaultSecret^*$ of degree smaller than $\SecretSize$ interpolating the pairs $(\VeatureElement_j,\VaultPoly(\VeatureElement_j))$ for $j=0,\ldots, \SecretSize-1$. If $\VaultPoly-\VaultSecret^*$ splits into distinct linear factors, output $\VeatureSet$ as the set of its roots; otherwise return \AlgFailure. 

\end{enumerate}
The probability that all elements $\VeatureElement_1,\ldots,\VeatureElement_{\SecretSize-1}$ are guessed correctly, is at least 
\begin{equation*}
\binom{\VeatureSize}{\WErrorSize} \binom{\VeatureSize-\WErrorSize-1}{\SecretSize-\WErrorSize-1}\binom{\FieldSize}{\SecretSize-1}^{-1} \leq \binom{\VeatureSize-1}{\SecretSize -1} \binom{\FieldSize}{\SecretSize-1}^{-1}, 
\end{equation*}
and with probability $\prob{\AlgA}$ algorithm $\AlgA$ outputs an $\VeatureElement_0 \in \VeatureSet  \cap \WeatureSet = \VeatureSet  \setminus \{\VeatureElement_1,\ldots,\VeatureElement_{\WErrorSize}\}$. The claim now follows directly from the lower bound (\ref{max_prob}) for the success probability $\prob{Single}$ of attacks on single vault records.  
\ENDPROOF

\fi

Note, that the condition $\NumOverlap> \VeatureSize -\SecretSize$ in Theorem \ref{Thm_partial_1} cannot be relaxed, because, as explained in Section \ref{sec:FullRecoveryAttack}, for $\NumOverlap \leq \VeatureSize -\SecretSize$, our partial recovery attack allows to completely recover the feature set $\VeatureSet$ with probability $1$.

For $\VeatureSize=\WeatureSize =\NumOverlap$, we have only one feature set and, thus, we obtain the following corollary which limits the success probability of partial recovery attacks on single vault records.
\begin{cor}
For fixed $\VeatureSize$, any algorithm that takes as input a vault record of a uniformly chosen feature set $\VeatureSet$ of size $\VeatureSize$ and outputs an element $\VeatureElement \in \VeatureSet$, has success probability $\prob{}\in \Oh(\FieldSize^{-1})$. Specifically, 
\begin{equation*}
\prob{} \, \leq \,\FieldSize^{\VeatureSize-\SecretSize}\, \frac{\binom{\FieldSize}{\SecretSize-1}} 
{\binom{\FieldSize}{\VeatureSize}\binom{\VeatureSize-1}{\SecretSize -1}}.
\end{equation*}
\end{cor}

The next result states that, for $\VeatureSize -\SecretSize \leq \NumOverlap < (t+k)/2$ or $\NumOverlap \leq \SecretSize$, no partial recovery attack can determine any element $\VeatureElement \in \VeatureSet$ with non-negligible probability for all field sizes. 
\iffull
The proof is given in the appendix. 
\fi

\begin{thm}\label{Thm_partial_2} 
Let $\VeatureSize\geq \WeatureSize$ and $(\VeatureSize+\SecretSize)/2 > \NumOverlap \geq \VeatureSize -\SecretSize$ or $\NumOverlap \leq \SecretSize$. Then, any algorithm that takes as input two vault records of uniformly chosen feature sets $\VeatureSet$ and $\WeatureSet$ of size $\VeatureSize$ and $\WeatureSize$, respectively, with $|\VeatureSet \cap \WeatureSet|=\NumOverlap$, and outputs an element $\VeatureElement \in (\VeatureSet \cup \WeatureSet)\setminus(\VeatureSet \cap \WeatureSet)$, has success probability $\prob{}\in \Oh(\FieldSize^{-1})$. 
\end{thm}

Unfortunately, we are not able to extend Theorem \ref{Thm_partial_2} to all $\NumOverlap < \VeatureSize-\SecretSize$. A reason for this is that, for $\VeatureSize \geq 3\SecretSize$, Theorem \ref{EntropyLoss} does not give positive bounds on the average min-entropy for any $\NumOverlap \in [2 \SecretSize,\VeatureSize - \SecretSize]$. Note that $\VeatureSize \geq 3\SecretSize$ implies $\VeatureSize - \SecretSize \geq (\VeatureSize+\SecretSize)/2$ and, thus, the claim of Theorem \ref{Thm_partial_2} becomes empty in that case. 

At least for $\SecretSize\geq\VeatureSize/2$, Theorem \ref{Thm_partial_2} covers all $\NumOverlap < (\VeatureSize + \SecretSize)/2$. Therefore, we obtain the following corollary. 

\begin{cor}
Let $\VeatureSize\geq \WeatureSize$, $\SecretSize\geq\VeatureSize/2$ and $\NumOverlap < (\VeatureSize+\SecretSize)/2$ . Then, any algorithm that takes as input two vault records of uniformly chosen feature sets $\VeatureSet$ and $\WeatureSet$ of size $\VeatureSize$ and $\WeatureSize$, respectively, with $|\VeatureSet \cap \WeatureSet|=\NumOverlap$, and outputs an element $\VeatureElement \in (\VeatureSet \cup \WeatureSet)\setminus(\VeatureSet \cap \WeatureSet)$, has success probability $\prob{}\in \Oh(\FieldSize^{-1})$. 
\end{cor}

%% file: preventions.tex
\section{Preventions}\label{sec:preventions}
In this section, we discuss possible preventions to avoid an adversary from running the partial recovery attack in order to cross-match the vaults or even recover the feature elements.

\subsection{Use the Original Fuzzy Vault Scheme}
The partial recovery attack (Algorithm \ref{Algo:partial_recovery}) presented in this paper makes use of the representation of the vault records as polynomials. Thus, the attack apparently cannot be applied to the original fuzzy vault scheme. However, to ensure that the original fuzzy vault resists correlation attacks \cite{bib:SimoensEtAl2009}, it is necessary to use a number of chaff points that is close to the maximum, \ie{} to set almost all unoccupied points as chaff; this measure however, implies a significant increase of memory each vault record consumes. 

\subsection{Protection by Additional Secrets}
An obvious countermeasure is to encrypt the vaults by a system-wide or user-specific secret key \cite{bib:BlantonAliasgari2012}, or to combine the features with a user password \cite{bib:NandakumarNagarJain2007}. However, with this measure, the security and robustness of the scheme relies on the secrecy and availability, respectively, of the key or password, exactly the constraint that biometric cryptosystems aim to remove \cite{bib:UludagEtAl2004}; in fact, the fuzzy vault scheme has been promoted as a key-less template protection scheme \cite{bib:KevenaarEtAl2010}. 

\subsection{Re-Ordering of the Feature Encodings}

Analogously to the random basis transformations suggested to prevent record multiplicity attacks against the fuzzy commitment scheme \cite{bib:KelkboomEtAl2011}, the (public) encoding of the feature space into the finite field $\FiniteField$ may be randomized. The encoding may be chosen record-specific or system-wide, depending on whether users may perform several \enrollments{} or not. The random encoding can be compared to a \emph{password salt} in that it does not need to be kept secret but aims at distorting the stored reference information. 

Obviously, if $\FieldSize \gg \VeatureSize$, randomly chosen encodings ensure that, with high probability, the representations (encodings) of
two arbitrary feature sets share only a small number of field elements. Hence, it is reasonable to assume that the random permutation between the two encodings destroys any algebraic similarities that could be exploited by attacks. However, further research is required to confirm this assumption.  

When using random record-specific encodings, a full code table must be stored along with the vault record. However, using pseudo-random functions, encodings can be efficiently generated from small seeds, \eg{} from counters, greatly reducing the storage requirements for the records. 

Due to its simplicity, effectiveness, and the absence of clear drawbacks, the random encoding is a very promising approach.

\subsection{Additional Randomness}
Roughly speaking, the partial recovery attack (Algorithm \ref{Algo:partial_recovery}) processes the upper coefficients of the improved fuzzy vault scheme. Therefore, an approach to thwart the attack is to add additional randomness to the upper coefficients. However, in order to allow recovery of the secret polynomial in genuine verification attempts and to prevent it in impostor attempts, the randomization should preserve the property that $\VaultPoly(\absc)=\VaultSecret(\absc)$ holds exactly for the genuine feature elements. This can be achieved by multiplying the characteristic polynomial $\CharPoly_{\VeatureSet}(\SYM)$ of $\VeatureSet$ with a random polynomial having no roots in $\FiniteField$, before adding it to the secret polynomial $\VaultPoly$. 

One method to efficiently choose such a polynomial is to select a random \emph{blending set} $\AdditionalVeatureSet$ with elements in $\FiniteFieldExt\setminus \FiniteField$, where $\FiniteFieldExt$ is an extension field of $\FiniteField$, and to use the characteristic polynomial $\CharPoly_\AdditionalVeatureSet$. Precisely, for a set $\VeatureSet$ of size $\VeatureSize$ of features encoded as elements in $\FiniteField$, select the secret polynomial $\VaultSecret\in\FiniteFieldExt[\SYM]$ of degree smaller than $\SecretSize$,\footnote{As before, $\VaultSecret$ can be chosen at random, resulting in a probabilistic version, or set to the negative of the lowest $\SecretSize$ terms of  $\CharPoly_{\AdditionalVeatureSet}\cdot\CharPoly_\VeatureSet$ to obtain a deterministic version with reduced record size.} generate a random blending set $\AdditionalVeatureSet$ of size $\AdditionalVeatureSize$ with elements from $\FiniteFieldExt\setminus \FiniteField$ and compute the vault record as the polynomial
\begin{align}
\begin{split}
\VaultPoly(\SYM)
&=\VaultSecret(\SYM)+\CharPoly_\AdditionalVeatureSet(\SYM)\cdot\CharPoly_\VeatureSet(\SYM).
\end{split}
\end{align}

To see how this randomization helps to prevent our attack, suppose that we are given a second vault with additional randomness $\WaultPoly=\WaultSecret+\CharPoly_\AdditionalWeatureSet\cdot\CharPoly_\WeatureSet$ where $\WaultSecret\in\FiniteFieldExt[\SYM]$ is of degree smaller than $\SecretSize$, $\WeatureSet\subset\FiniteField$ encodes $\WeatureSize\leq\VeatureSize$ features, and randomly chosen $\AdditionalWeatureSet\subset\FiniteFieldExt\setminus\FiniteField$ of size $\AdditionalWeatureSize$. For simplicity, we assume that $\AdditionalWeatureSize+\WeatureSize=\AdditionalVeatureSize+\VeatureSize$ such that $\VaultPoly$ and $\WaultPoly$ are of the same degree. To successfully apply Algorithm \ref{Algo:partial_recovery} (in the extension field), the requirement 
\begin{equation}
|\AdditionalVeatureSet\cap\AdditionalWeatureSet|+\NumOverlap\geq(\AdditionalVeatureSize+\VeatureSize+\SecretSize)/2,
\label{eq:DecAttackRequirement}
\end{equation}
with $\NumOverlap=|\VeatureSet\cap\WeatureSet|$, must be fulfilled; the requirement $\NumOverlap\geq(\VeatureSize+\SecretSize)/2$ is neither sufficient nor necessary for (\ref{eq:DecAttackRequirement}).

Since $\NumOverlap\leq\VeatureSize$, the requirement (\ref{eq:DecAttackRequirement}) is fulfilled only if  
\begin{equation*}
|\AdditionalVeatureSet\cap\AdditionalWeatureSet|\geq(\AdditionalVeatureSize-\VeatureSize+\SecretSize)/2.
\end{equation*}
This yields a criterion how the new parameters $\AdditionalVeatureSize$ and $|\FiniteFieldExt|$ has to be chosen to effectively thwart the partial recovery attack.

The described modification increases the size of the vault records from $\VeatureSize \log |\FiniteField|$ to $(\VeatureSize + \AdditionalVeatureSize)\log |\FiniteFieldExt|$ (in the probabilistic version).

%% file: discussion.tex
\section{Discussion}\label{sec:discussion}
In this paper, we presented a new attack via record multiplicity against the improved fuzzy vault scheme of Dodis \etal{} \cite{bib:DodisEtAl2004,bib:DodisEtAl2008}. Using the extended Euclidean algorithm, our attack links two vault records of the same individual with high probability (equal to the genuine acceptance rate for systems based on a Reed-Solomon decoder) and uncovers the elements, \eg{} the biometric data, in which the feature sets protected by the records differ. The algorithm is very efficient running in random polynomial time and can be proved to succeed provided that the two feature sets overlap sufficiently. We thereby disproved the conjecture of Blanton and Aliasgari \cite{bib:BlantonAliasgari2013} that solving the equations given by two related vault records is computationally infeasible if the parameters are not small. The experiments we conducted empirically verified the effectiveness and reliability of the attack for parameters that can be expected in practice; in particular, we demonstrated that, typically, the probability of spurious outputs, \ie{} false positives, is very small.

For certain parameters, specifically for $\VeatureSize \gg \SecretSize$ and $\NumOverlap\geq(\VeatureSize+\SecretSize)/2$ not too large, the recovery of the differing feature elements already allows full recovery of the protected feature sets. For the general case, we showed how the partial recovery of the feature by our attack can be extended to a probabilistic full recovery method, resulting in a significantly improved brute-force attack as compared to previous approaches \cite{bib:MihailescuMunkTams2009}. 

\iffull
On the other hand, we derived bounds for the information leakage from multiple vault records, which allows to bound the success probability of attacks via record multiplicity. Using these bounds we were able to show that our partial and full recovery attacks are asymptotically optimal for a wide range of parameters, and could even prove the general result that related vault records do not reveal any elements from the overlap of the protected feature sets.   
\fi

Our work confirms that additional countermeasures are necessary in order to allow secure reuse of the improved fuzzy vault scheme. We presented several approaches to thwart record multiplicity attacks without modifying the general properties of the scheme; in particular, (pseudo-)random embedding of the features into the finite field seems to be a promising approach. However, further research is needed to give theoretical evidence for its effectiveness. 

As a final remark, we stress that, generally, attacks against biometric cryptosystems should not only be considered from an information theoretic but also from a complexity theoretic point of view, which is currently of greater practical relevance for modalities such as fingerprints (\eg{} see \cite{bib:MerkleEtAl2010ProvableSecurity}).

%% file: appendix.tex
\begin{appendix}

\ifarxiv
\section{Appendix}
\fi

\input{proof1}

\iffull
\input{proof2}

\input{proof3}

\fi

\end{appendix}

%% file: proof1.tex
\subsection{Proof of Theorem \ref{thm:main}}\label{sec:ProofMain}

\subsubsection*{Proof of a)}
We show that Algorithm \ref{Algo:partial_recovery} outputs a correct triple if there exist feature sets $\VeatureSet$ and $\WeatureSet$ protected by $\VaultPoly$ and $\WaultPoly$ sharing at least $(\VeatureSize+\SecretSize)/2$ elements. Therefore, let $\NumOverlap=|\VeatureSet\cap\WeatureSet|$ and set
\begin{align}
\begin{split}
\VErrorPoly(\SYM)&=\CharPoly_{\WeatureSet\setminus\VeatureSet}(\SYM)=\prod_{\absc\in\WeatureSet\setminus\VeatureSet}(\SYM-\absc)\quad\text{and}\\
\WErrorPoly(\SYM)&=\CharPoly_{\VeatureSet\setminus\WeatureSet}(\SYM)=\prod_{\absc\in\VeatureSet\setminus\WeatureSet}(\SYM-\absc)
\label{eq:ErrorPolys}
\end{split}
\end{align}
Furthermore, we write
\begin{equation*}
\WErrorSize=\deg(\WErrorPoly)\quad\text{and}\quad\VErrorSize=\deg(\VErrorPoly).
\end{equation*}
We prove the desired statement with the following.
\begin{lemma}
\label{lemma:main}
Suppose that $\NumOverlap\geq(\VeatureSize+\SecretSize)/2$ and let 
\begin{equation*}
\RemPoly_j=\VErrorPoly_j\cdot\VaultPoly+\WErrorPoly_j\cdot\WaultPoly
\end{equation*} 
be the sequence in the extended Euclidean algorithm applied to $\VaultPoly$ and $\WaultPoly$ (\eg{} Algorithm 3.6 in \cite{bib:vzGathen2003}). Set 
\begin{equation*}
\CandidateIndicesList = \{~j~|~0 \leq \deg(\RemPoly_j) < \deg(\WErrorPoly_j)+\SecretSize~\}
\end{equation*}
and $j_0 = \argmin_{j \in\CandidateIndicesList} \left(\deg(\WErrorPoly_j) \right)$.
Then $\VErrorPoly=\EEAConstant\cdot\VErrorPoly_{j_0}$ and $\WErrorPoly=-\EEAConstant\cdot\WErrorPoly_{j_0}$ for a non-zero $\EEAConstant\in \FiniteField$.
\end{lemma}
To show that the lemma in fact holds, we apply a lemma of Gao (Lemma 3.2 in \cite{bib:Gao2002}) which we state first using our notation such it easily applies to our situation. 
\begin{lemma}[Gao 2002]
Let $\VaultPoly(\SYM)=\CharPoly(\SYM)\VErrorPoly(\SYM)+\VaultSecret(\SYM)$ and $\WaultPoly(\SYM)=\CharPoly(\SYM)\WErrorPoly(\SYM)+\WaultSecret(\SYM)$, with $\gcd(\VErrorPoly,\WErrorPoly)=1$ and
\begin{equation*}
\deg(\VErrorPoly),\deg(\WErrorPoly)\leq\WErrorSize,\quad\quad\deg(\VaultSecret),\deg(\WaultSecret)\leq\SecretSize-1.
\end{equation*}
Suppose that $\NumOverlap$ satisfies
\begin{equation*}
\deg(\CharPoly)\geq\NumOverlap>\SecretSize-1+\WErrorSize.
\end{equation*}
Apply the extended Euclidean algorithm to $\VaultPoly$ and $\WaultPoly$ to obtain the sequence
\begin{equation*}
\RemPoly_j = \VErrorPoly_j\cdot\VaultPoly+\WErrorPoly_j\cdot\WaultPoly.
\end{equation*}
Let $j_{\NumOverlap}$ be minimal such that $\deg(\RemPoly_{j_{\NumOverlap}})<\NumOverlap$. Then there exists a non-zero $\EEAConstant\in\FiniteField$ with $\VErrorPoly_{j_{\NumOverlap}}=\EEAConstant\cdot\VErrorPoly$ and $\WErrorPoly_{j_{\NumOverlap}}=-\EEAConstant\cdot\WErrorPoly$.
\label{lemma:Gao}
\end{lemma}
\BEGINPROOF[Proof of Lemma \ref{lemma:main}]
We first verify that the requirements of Lemma \ref{lemma:Gao} are fulfilled for $\CharPoly=\CharPoly_{\VeatureSet\cap\WeatureSet}$. Therefore, note that the polynomials $\VErrorPoly$ and $\WErrorPoly$ are in fact co-prime. Furthermore, $\WErrorSize\geq\VErrorSize$ which is implied by the relation $\VeatureSize+\VErrorSize=\WeatureSize+\WErrorSize$; thus, $\WErrorSize$ is an upper bound for $\deg(\VErrorPoly)$ and $\deg(\WErrorPoly)$. Note that $\WErrorSize=\VeatureSize-\NumOverlap$ and that $2\NumOverlap\geq\VeatureSize+\SecretSize$. Hence
\begin{equation*}
\label{eq:IsSmallerThanNumOverlap}
\SecretSize-1+\WErrorSize = \SecretSize+\VeatureSize-\NumOverlap-1 < \NumOverlap
\end{equation*}
and thus the requirements of Lemma \ref{lemma:Gao} are fulfilled.

Consequently, there exists a $j_{\NumOverlap}$ with $\VErrorPoly_{j_{\NumOverlap}}=\EEAConstant\cdot\VErrorPoly$ and $\WErrorPoly_{j_{\NumOverlap}}=-\EEAConstant\cdot\WErrorPoly$ for a non-zero field element $\EEAConstant$, and $\deg(\RemPoly_{j_{\NumOverlap}}) < \NumOverlap$. It follows that $\VErrorPoly_{j_{\NumOverlap}} \cdot\CharPoly_{\VeatureSet} + \WErrorPoly_{j_{\NumOverlap}}\cdot \CharPoly_{\WeatureSet} = 0$ and thus, 
\begin{equation*}
\RemPoly_{j_{\NumOverlap}} = \EEAConstant\left( \VErrorPoly \cdot \VaultSecret - \WErrorPoly \cdot\WaultSecret \right), 
\end{equation*}
which implies $j_{\NumOverlap}\in\CandidateIndicesList$ (and, in particular, that $\CandidateIndicesList$ is not empty). 

By definition of $j_0$, we have $\deg(\WErrorPoly_{j_0}) \leq \deg(\WErrorPoly_{j_{\NumOverlap}}) < \VeatureSize-\NumOverlap$. On the other hand, since $\deg(\RemPoly_{j_0}) < \VeatureSize$, both summands $\VErrorPoly_{j_0}\cdot\VaultPoly$ and $\WErrorPoly_{j_0}\cdot\WaultPoly$ of $\RemPoly_{j_0}$ must have the same degree and thus $\deg(\VErrorPoly_{j_0}) \leq \deg(\WErrorPoly_{j_0}) < \VeatureSize-\NumOverlap$. Furthermore, since $j_0 \in\CandidateIndicesList$, we have $\deg(\RemPoly_{j_0})< \deg(\WErrorPoly_{j_0})+\SecretSize$. Therefore, both the polynomial 
\begin{equation*}
\RelationErrorPoly = \VErrorPoly_{j_0} \cdot \VaultSecret - \WErrorPoly_{j_0} \cdot\WaultSecret
\end{equation*}
and $\RemPoly_{j_0}$ have degree strictly smaller than $\VeatureSize - \NumOverlap + \SecretSize$, which is at most $\NumOverlap$.  

On the other hand, 
\begin{align}
\begin{split}
\RemPoly_{j_0} &= \VErrorPoly_{j_0}\cdot\VaultPoly+\WErrorPoly_{j_0}\cdot\WaultPoly \\
 &= \VErrorPoly_{j_{0}} \cdot\CharPoly_{\VeatureSet} + \WErrorPoly_{j_{0}}\cdot \CharPoly_{\WeatureSet} + \RelationErrorPoly \\
 &= \CharPoly_{\VeatureSet\cap\WeatureSet} \left( \VErrorPoly_{j_{0}} \cdot\WErrorPoly + \WErrorPoly_{j_{0}}\cdot \VErrorPoly\right) + \RelationErrorPoly
\label{Split_Eukl_Rel}
\end{split}
\end{align}
Since $\deg(\CharPoly_{\VeatureSet\cap\WeatureSet})=\NumOverlap$ and  $\deg(\RemPoly_{j_0}),\deg(\RelationErrorPoly)<\NumOverlap$, we can conclude from (\ref{Split_Eukl_Rel}) that $\VErrorPoly_{j_{0}} \cdot\WErrorPoly + \WErrorPoly_{j_{0}}\cdot \VErrorPoly=0$. As $\VErrorPoly$ and $\WErrorPoly$ are co-prime, this implies that $\VErrorPoly_{j_{0}}=\EEAConstant' \VErrorPoly$ and $\WErrorPoly_{j_{0}}=-\EEAConstant' \WErrorPoly$. On the other hand, $\deg(\WErrorPoly_{j_{0}})\leq \deg(\WErrorPoly_{j_{\NumOverlap}})=\deg( \WErrorPoly)$ and, consequently, $\EEAConstant' \in \FiniteField$. By definition of $\CandidateIndicesList$, $\RemPoly_{j_0}\neq 0$ and thus, $\EEAConstant'$ is non-zero. 
This proves the statement of Lemma \ref{lemma:main}.
\ENDPROOF

Consider the following.
\begin{cor}\label{cor:EfficientRecovery}
Suppose that $\NumOverlap\geq(\VeatureSize+\SecretSize)/2$ and $\RemPoly_{j_0}=\VErrorPoly_{j_0}\cdot\VaultPoly+\WErrorPoly_{j_0}\cdot\WaultPoly$ are as in Step \ref{partial_index} of Algorithm \ref{Algo:partial_recovery}. Then
\begin{equation*}
\VaultSecret\equiv\VaultPoly\mod\WErrorPoly_{j_0}.
\end{equation*}
\end{cor}
\BEGINPROOF
By Lemma \ref{lemma:main} we see that $\WErrorPoly_{j_0}$ divides $\CharPoly_{\VeatureSet}$ and thus the corollary follows.
\ENDPROOF
In particular, assuming $\NumOverlap\geq(\VeatureSize+\SecretSize)/2$, Algorithm \ref{Algo:partial_recovery} passes the test made in Step \ref{partial_rem} and together with Lemma \ref{lemma:main} the desired Statement a) of Theorem \ref{thm:main} follows.

\subsubsection*{Proof of b)}
To prove Statement b), note that the polynomials $\VErrorPoly_{j}$ and $\WErrorPoly_{j}$ in Algorithm \ref{Algo:partial_recovery} must be co-prime; this is a necessary property of the coefficients in the extended Euclidean algorithm (\eg{} see the proof of Lemma 3.8 in \cite{bib:vzGathen2003}). Consider the following.
\begin{lemma}\label{lemma:b_implies_a}
Assume $\SecretSize\leq\WeatureSize\leq\VeatureSize$ and let $\VaultPoly,\WaultPoly\in\FiniteField[\SYM]$ be of degree $\VeatureSize$ and $\WeatureSize$, respectively. Suppose that there exist co-prime polynomials $\VErrorPoly^*,\WErrorPoly^*\in\FiniteField[\SYM]$ with $\WErrorPoly^*\neq 0$ such that $\RemPoly=\VErrorPoly^*\cdot\VaultPoly+\WErrorPoly^*\cdot\WaultPoly$ is non-zero with degree smaller than $\deg(\WErrorPoly^*)+\SecretSize$; furthermore, assume that $\hat{\VaultSecret}=\VaultPoly\rem\WErrorPoly^*$ is of degree smaller than $\SecretSize$.\footnote{By $\VaultPoly\rem\WErrorPoly^*$ we denote the remainder of $\VaultPoly$ divided by $\WErrorPoly^*$.}

Then there exist $\hat{\WaultSecret},\CharPoly\in\FiniteField[\SYM]$ and a non-zero $\EEAConstant\in\FiniteField$ such that:
\begin{enumerate}
\item $\deg(\hat{\WaultSecret})<\SecretSize$;
\item $\VaultPoly=\hat{\VaultSecret}+\CharPoly\cdot\hat{\WErrorPoly}$ and $\WaultPoly=\hat{\WaultSecret}+\CharPoly\cdot\hat{\VErrorPoly}$ with $\hat{\VErrorPoly}=\EEAConstant\cdot\VErrorPoly^*$ and $\hat{\WErrorPoly}=-\EEAConstant\cdot\WErrorPoly^*$;
\item $\deg(\CharPoly)=\VeatureSize-\deg(\WErrorPoly^*)$.
\end{enumerate}

\end{lemma}
\begin{proof}
From $\RemPoly=\VErrorPoly^*\cdot\VaultPoly+\WErrorPoly^*\cdot\WaultPoly$ and $\deg(\RemPoly)<\deg(\WErrorPoly^*)+\SecretSize\leq \deg(\WErrorPoly^*)+\WeatureSize$, we can conclude $\deg(\VErrorPoly^*)+\VeatureSize=\deg(\WErrorPoly^*)+\WeatureSize$ and, in particular, $\deg(\VErrorPoly^*)\leq\deg(\WErrorPoly^*)$.

Set $\hat{\WaultSecret}=\VErrorPoly^*\cdot(\VaultPoly-\hat{\VaultSecret})/\WErrorPoly^*+\WaultPoly$. Since $\WErrorPoly^*$ divides $\VaultPoly-\VaultSecret$, $\hat{\WaultSecret}$ is a (non-fractional) polynomial in $\FiniteField[\SYM]$. It follows that
\begin{equation}\label{lem4_eq_PfQg}
\VErrorPoly^*\cdot\hat{\VaultSecret}+\WErrorPoly^*\cdot\hat{\WaultSecret}=\VErrorPoly^*\cdot\VaultPoly+\WErrorPoly^*\cdot\WaultPoly=\RemPoly.
\end{equation}
Since, by assumption, $\deg(\RemPoly)<\deg(\WErrorPoly^*)+\SecretSize$, $\deg(\hat{\VaultSecret})<k$, and $\deg(\VErrorPoly^*)\leq\deg(\WErrorPoly^*)$, we may conclude from (\ref{lem4_eq_PfQg}) that $\deg(\hat{\WaultSecret})<\SecretSize$. 

Set $\VaultSpurChar=\VaultPoly-\hat{\VaultSecret}$ and $\WaultSpurChar=\WaultPoly-\hat{\WaultSecret}$. Then
\begin{equation*}
\VErrorPoly^*\cdot\VaultSpurChar+\WErrorPoly^*\cdot\WaultSpurChar=0.
\end{equation*}
Furthermore, since $\deg(\hat{\VaultSecret}),\deg(\hat{\WaultSecret})<\SecretSize\leq\WeatureSize\leq\VeatureSize$, $\VaultSpurChar$ and $\WaultSpurChar$ must be of degree $\VeatureSize$ and $\WeatureSize$, respectively.

Write $\CharPoly=\gcd(\VaultSpurChar,\WaultSpurChar)$, $\hat{\WErrorPoly}=\VaultSpurChar/\CharPoly$, and $\hat{\VErrorPoly}=\WaultSpurChar/\CharPoly$.  Note that
\begin{equation*}
0=\VErrorPoly^*\cdot\VaultSpurChar+\WErrorPoly^*\cdot\WaultSpurChar=\CharPoly\cdot\underbrace{(\VErrorPoly^*\cdot\hat{\WErrorPoly}+\WErrorPoly^*\cdot\hat{\VErrorPoly})}_{=0}
\end{equation*}
Since, by construction, $\hat{\VErrorPoly}$ and $\hat{\WErrorPoly}$ are co-prime, we may conclude that there exists a non-zero 
$\EEAInverse\in\FiniteField[\SYM]$ with $\VErrorPoly^*=\EEAInverse\cdot\hat{\VErrorPoly}$ and $\WErrorPoly^*=-\EEAInverse\cdot\hat{\WErrorPoly}$. Since $\VErrorPoly^*$ and $\WErrorPoly^*$ are co-prime, $\EEAInverse \in \FiniteField$ and, hence, setting $\EEAConstant=\EEAInverse^{-1}$ yields the statement of the lemma.
\end{proof} 

Now, Statement b) follows as a corollary. More specifically, if Algorithm \ref{Algo:partial_recovery} outputs $(\NumOverlap^*,\VeatureSet_0,\WeatureSet_0)$, then there exist polynomials $\RemPoly,\VErrorPoly^*,\WErrorPoly^*$ such that the requirements in Lemma \ref{lemma:b_implies_a} are fulfilled, \ie{} the polynomials $\RemPoly_{j_0}$, $\VErrorPoly_{j_0}$ and $\WErrorPoly_{j_0}$ such that $\WErrorPoly^*$ passes the test made in Step \ref{partial_rem} of Algorithm \ref{Algo:partial_recovery}. It follows with the lemma that there exists a polynomial $\CharPoly\in\FiniteField[\SYM]$ of degree $\NumOverlap^*=\VeatureSize-\deg(\WErrorPoly^*)$ and polynomials $\hat{\VaultSecret},\hat{\WaultSecret}\in\FiniteField[\SYM]$ of degree smaller than $\SecretSize$ such that $\VaultPoly=\hat{\VaultSecret}+\CharPoly\CharPoly_{\VeatureSet_0}$ and $\WaultPoly=\hat{\WaultSecret}+\CharPoly\CharPoly_{\WeatureSet_0}$.

\subsubsection*{Proof of c)}
In this section, we estimate the running time of the partial recovery attack. For the running times of the used sub-algorithms, we refer to \cite{bib:vzGathen2003}.

\begin{lemma}\label{lemma:RunningTime}
Algorithm \ref{Algo:partial_recovery} can be implemented using an expected number of 
\begin{equation*}
\Oh(\VeatureSize^2)+\SoftOh(\VeatureSize\cdot\log\FieldSize)
\end{equation*}
operations in $\FiniteField$.
\end{lemma}
\begin{proof}
In Step \ref{partial_eea} of the algorithm the sequence in the extended Euclidean algorithm is computed which consumes $\Oh(\VeatureSize^2)$ finite field operations and which dominates the time needed to compute the remainder in Step \ref{partial_rem}. If an index $j_0$ is found in Step \ref{partial_index}, the polynomials' root have to be computed in Step \ref{partial_roots}. This can be performed using an expected number of $\SoftOh(\VeatureSize\cdot\log|\FiniteField|)$ operations in $\FiniteField$ (Corollary 14.16 in \cite{bib:vzGathen2003}). Thus, Algorithm \ref{Algo:partial_recovery} can be implemented using the claimed number of operations in $\FiniteField$.
\end{proof}

The above lemma states that the partial recovery attack can be implemented to run in non-deterministic polynomial time. It is in fact not known whether the algorithm can be implemented to run in deterministic polynomial time. The question whether there exists a deterministic polynomial-time algorithm depends on whether there exists a deterministic polynomial-time algorithm for finding the roots of a polynomial in a finite field. For a variety of special cases, assuming the \emph{extended Riemann hypothesis}, it is known that deterministic algorithms factoring polynomials whose running times are bounded by a polynomial exists. On the other hand, factoring polynomials can be performed very fast using very efficient randomized algorithms. For a survey as well as further references we refer to \cite{bib:vzGathenPanario2001}.

%% file: proof2.tex
\subsection{Proof of Theorem \ref{EntropyLoss}}
In order to prove the theorem, we need an estimation for the number of values that a pair $(\VaultPoly, \WaultPoly)$ of vault records can take, so that $\VaultPoly$ and $\WaultPoly$ correspond to (\ie{} can be constructed from) feature sets $\VeatureSet$ and $\WeatureSet$, respectively, with given set distance $\SetDiff \left(\VeatureSet,\WeatureSet\right)=\diffVW$. We do this by means of the following lemma which limits the number of values for the upper $\GenuineSize-\SecretSize$ coefficients of the vault records, \ie{} the number of vault records in the deterministic version of the improved fuzzy vault. 

\begin{lemma}\label{Lemma_Loss}\label{lem_Entropy}
For a given deterministic vault record $\VaultPoly=(\VaultCoeff_\SecretSize,\ldots,\VaultCoeff_{\VeatureSize})$ let $\VaultSurr^{\WeatureSize}_{\diffVW}(\VaultPoly)$ denote the set of deterministic vault records $\WaultPoly=(\WaultCoeff_\SecretSize,\ldots, \WaultCoeff_{\WeatureSize})$ so that $\VaultPoly$ and $\WaultPoly$ can be constructed from feature sets $\VeatureSet$ and $\WeatureSet$, respectively, with $\SetDiff \left(\VeatureSet,\WeatureSet\right)=\diffVW$. 
Then for any vault record $\VaultPoly=(\VaultCoeff_\SecretSize,\ldots,\VaultCoeff_{\VeatureSize})$, $\WeatureSize \leq \VeatureSize$ and $0\leq \diffVW \leq \WeatureSize + \VeatureSize$, the cardinality of $\VaultSurr^{\WeatureSize}_{\diffVW }(\VaultPoly)$ is limited by $\FieldSize^{\diffVW}$.\end{lemma}

\BEGINPROOF We first show the result for the case $\VeatureSize=\WeatureSize$ and afterwards generalize it to arbitrary feature set sizes. Both cases, \ie{} $\VeatureSize=\WeatureSize$ and the general case, are proven by induction. 

\paragraph{The case of equal feature set sizes ($\mathbf{\VeatureSize=\WeatureSize}$).} Note that the set difference between sets of equal size is always even. For $\diffVW = 0$, the claim is trivial. Let $\VeatureSet,\WeatureSet$ be two feature sets of size $\VeatureSize$ with  $\diffVW = \SetDiff \left(\VeatureSet,\WeatureSet\right)=2$, \ie{} $\VeatureSet=(\VeatureSet \cap \WeatureSet) \cup \{\VeatureElement\}$ and $\WeatureSet=(\VeatureSet \cap \WeatureSet) \cup \{\WeatureElement\}$. 

Since the vault records $\VaultPoly$ and $\WaultPoly$ are the coefficients of the characteristic polynomial of $\VeatureSet$ and $\WeatureSet$, respectively, they can be written as $\VaultCoeff_m=\sigma_{\VeatureSize - m}(\VeatureSet)$ and $\WaultCoeff_m=\sigma_{\VeatureSize - m}(\WeatureSet)$ for $m=\SecretSize,\ldots,\VeatureSize$, 
where $\sigma_{m}$, as defined in (\ref{eq:elsympol}), denotes the $m$-th elementary symmetric polynomial \cite{bib:LidlNiederreiter}. 

From the general equation 
\ifdouble
\begin{align}
\begin{split}
\sigma_{m}(\SYM_1,\ldots,\SYM_n)=&\SYM_1 \cdot \sigma_{m-1}(\SYM_2,\ldots,\SYM_n) \nonumber \\
&\; + \sigma_{m}(\SYM_2, \ldots, \SYM_n)\nonumber 
\end{split}
\end{align}

\else
\begin{equation*}
\sigma_{m}(\SYM_1,\ldots,\SYM_n)=\SYM_1 \cdot \sigma_{m-1}(\SYM_2,\ldots,\SYM_n) + \sigma_{m}(\SYM_2, \ldots, \SYM_n)
\end{equation*}
\fi
follows that $\VaultCoeff_m - \WaultCoeff_m  = (\VeatureElement - \WeatureElement ) \cdot \sigma_{\VeatureSize - m - 1}(\VeatureSet \cap \WeatureSet)$ and by recursively using $\sigma_{r}(\VeatureSet \cap \WeatureSet) = \sigma_{r}(\VeatureSet) - \VeatureElement \cdot \sigma_{r-1}(\VeatureSet \cap \WeatureSet)$ for $r=\VeatureSize - m - 1, \ldots, 1$ as well as $\sigma_{0}(\VeatureSet \cap \WeatureSet) =  \sigma_{0}(\VeatureSet) = 1$, we obtain 
\begin{eqnarray*}
\VaultCoeff_m - \WaultCoeff_m & = & (\VeatureElement - \WeatureElement ) \cdot \sum_{i=1}^{\VeatureSize - m} (-\VeatureElement)^{i-1} \sigma_{\VeatureSize - m - i}(\VeatureSet)
 \\
 & = & (\VeatureElement - \WeatureElement) \cdot \sum_{i=1}^{\VeatureSize-m} (-\VeatureElement)^{i-1} \VaultCoeff_{m+i}
\end{eqnarray*}
for $m=\SecretSize,\ldots,\VeatureSize$. Consequently, the elements $\WaultCoeff_{\SecretSize+1},\ldots,\WaultCoeff_{\VeatureSize}$ are uniquely determined by $\VeatureElement$, $\WeatureElement$ and the vault record $\VaultPoly$. This proves the result for $\diffVW=2$ (in the case $\VeatureSize=\WeatureSize$). 

We now estimate $\left| \VaultSurr^{\VeatureSize}_{\diffVW}(\VaultPoly)\right |$ for $\diffVW > 2$ and assume that $\left| \VaultSurr^{\VeatureSize}_{i}(Z)\right | \leq q^{i}$ holds for any vault record $\ZaultPoly=(\ZaultCoeff_{\SecretSize}, \ldots,\ZaultCoeff_{\VeatureSize})$ and for all even $i < \diffVW$.  

For any pair of feature sets $\VeatureSet$ and $\WeatureSet$ of size $\VeatureSize$ having set difference $\SetDiff \left(\VeatureSet,\WeatureSet\right)=\diffVW$, there exists a set $\ZeatureSet$ of size $\VeatureSize$ with $\SetDiff \left(\WeatureSet,\ZeatureSet\right)=\diffVW-2$ and $\SetDiff \left(\VeatureSet,\ZeatureSet\right)=2$. Thus, 
\begin{equation*}
\VaultSurr^{\VeatureSize}_{\diffVW}(\VaultPoly) \subseteq \bigcup_{\ZeatureSet \in \, \VaultSurr^{\VeatureSize}_{2}(\VaultPoly)} \VaultSurr^{\VeatureSize}_{\diffVW-2}(\ZeatureSet).
\end{equation*}
Therefore, by applying the induction assumption for $i=2$ and $i=\diffVW-2$, we obtain
\begin{eqnarray*}
\left | \VaultSurr^{\VeatureSize}_{\diffVW}(\VaultPoly) \right| & \leq & \sum_{\ZeatureSet \in \, \VaultSurr^{\VeatureSize}_{2}(\VaultPoly)} \left | \VaultSurr^{\VeatureSize}_{\diffVW - 2}(\ZeatureSet)\right | \\[1ex]
& \leq & \sum_{\ZeatureSet \in \, \VaultSurr^{\VeatureSize}_{2}(\VaultPoly)} \FieldSize^{\diffVW - 2} \\[1ex]
&  \leq &  \FieldSize^{\diffVW}. 
\end{eqnarray*}

\paragraph{The case of unequal feature set sizes ($\mathbf{\VeatureSize>\WeatureSize}$).}
Let $\diffVW=1$, \ie{} let $\VeatureSet,\WeatureSet$ be two feature sets of size $\VeatureSize$ and $\VeatureSize-1$, respectively, with $\diffVW = \SetDiff \left(\VeatureSet,\WeatureSet\right)=1$, \ie{} $\VeatureSet=\WeatureSet \cup \{\VeatureElement\}$. From Lemma \ref{lem:coeff_char_poly}, we see that 
\begin{equation*}
\WaultCoeff_m = \sum_{i=m+1}^{\VeatureSize} \VeatureElement^{i-m-1} \VaultCoeff_{i}.
\end{equation*}  
for $m=\SecretSize,\ldots,\VeatureSize-1$. Thus, $\WaultPoly$ is uniquely determined by $\VaultPoly$ and $\VeatureElement$, which proves the result for $\diffVW=1$.

We now estimate $\left| \VaultSurr^{\WeatureSize}_{\diffVW}(\VaultPoly)\right |$ for $\diffVW > 1$ and $\WeatureSize < \VeatureSize$. By induction assumption, we can assume that $\left| \VaultSurr^{\WeatureSize}_{i}(\ZeatureSet)\right | \leq \FieldSize^{i}$ holds for any vault record $\ZeatureSet=(\ZeatureElement_{\SecretSize}, \ldots, \ZeatureElement_{n})$ with $\WeatureSize < n\leq \VeatureSize$ and all $i < \diffVW$. 

Analogously to the case for equal feature set sizes, we argue that for given feature set $\VeatureSet$ of size $\VeatureSize$, all feature sets $\WeatureSet$ of size $\WeatureSize\leq \VeatureSize$ having set difference $\SetDiff \left(\VeatureSet,\WeatureSet\right)=\diffVW$ can be reached by enumerating all sets $\WeatureSet$ of size $\WeatureSize$ with $\SetDiff \left(\WeatureSet,\ZeatureSet\right)=\diffVW-1$ for all sets $\ZeatureSet$ of size $\VeatureSize-1$ with $\SetDiff \left(\VeatureSet,\ZeatureSet\right)=1$. Thus, we get 
\begin{equation*}
\left |\VaultSurr^{\WeatureSize}_{\diffVW}(\VaultPoly) \right |= \sum_{Z \in \, \VaultSurr^{\VeatureSize-1}_{1}(\VaultPoly)} \left |\VaultSurr^{\WeatureSize}_{\diffVW - 1}(Z) \right |.
\end{equation*}
Since the vault records $Z$ in the sum above are of degree $\VeatureSize-1$, \ie{} are computed from feature sets of size $\VeatureSize-1$, we have to distinguish two cases. If $\WeatureSize<\VeatureSize-1$, we get $\left |\VaultSurr^{\WeatureSize}_{\diffVW - 1}(Z) \right | \leq q^{\diffVW -1}$ by induction assumption. On the other hand, if 
$\WeatureSize=\VeatureSize-1$, we get the same bound from the result for equal feature set sizes (see first part of the proof). Thus, in any case, we obtain
\begin{eqnarray*}
\left | \VaultSurr^{\WeatureSize}_{\diffVW}(\VaultPoly) \right| & \leq & 
\sum_{Z \in \, \VaultSurr^{\VeatureSize-1}_{1}(\VaultPoly)} \FieldSize^{\diffVW - 1}, \\[1ex]
 & \leq &  \FieldSize^{\diffVW}.
\end{eqnarray*}
\ENDPROOF

Using Lemma \ref{lem_Entropy}, can now prove Theorem \ref{EntropyLoss}. 

\BEGINPROOF[Proof of Theorem \ref{EntropyLoss}]
Considering the secure sketch (see \cite{bib:DodisEtAl2004} for a definition) that takes as input $(\VeatureSet,\WeatureSet)$ and outputs the corresponding vaults $\VaultPoly$ and $\WaultPoly$ as {\em secure sketch} (SS) we can apply Lemma 2.2 (b) in \cite{bib:DodisEtAl2004} to obtain
\begin{equation*}
\mathbf{\tilde{H}_{\infty}}(\VeatureSet,\WeatureSet | \VaultPoly,\WaultPoly) \geq \mathbf{H_{\infty}}(\VeatureSet,\WeatureSet,\VaultPoly,\WaultPoly) - \lambda, 
\end{equation*}
where $2^\lambda$ is the number of values the output $(\VaultPoly, \WaultPoly)$ can take, \ie{} the number of pairs of vaults that can be generated from feature sets $\VeatureSet$ and $\WeatureSet$ with set difference $\diffVW$. 

In the case of the deterministic version of the improved fuzzy vault, $\VaultPoly$ and $\WaultPoly$ contain $\VeatureSize - \SecretSize$ and $\WeatureSize - \SecretSize$ elements, respectively, so that $\lambda \leq (\VeatureSize + \WeatureSize - 2\SecretSize)\log \FieldSize$. On the other hand, $\WaultPoly \in \VaultSurr^{\WeatureSize}_{\diffVW}(\VaultPoly)$ and, hence, from Lemma \ref{lem_Entropy} we get $\lambda \leq (\VeatureSize - \SecretSize + \diffVW) \log \FieldSize$. Furthermore, since no randomness is added to the deterministic vaults, we have $\mathbf{H_{\infty}}(\VeatureSet,\WeatureSet,\VaultPoly,\WaultPoly)=\mathbf{H_{\infty}}(\VeatureSet,\WeatureSet)$. This gives the desired result. 

In the case of the probabilistic version of the fuzzy vault, $2 \SecretSize \log \FieldSize$ bits of entropy are added to $\mathbf{H_{\infty}}(\VeatureSet,\WeatureSet,\VaultPoly,\WaultPoly)$ by the random polynomials, but this additional term cancels out with the increase of $\lambda$ by $2 \SecretSize \log \FieldSize$ introduced by the $2 \SecretSize$ additional coefficients $\VaultCoeff_{0}, \ldots, \VaultCoeff_{\SecretSize-1}, \WaultCoeff_{0}, \ldots, \WaultCoeff_{\SecretSize-1}$ in the output $(\VaultPoly, \WaultPoly)$ of the secure sketch. 
\ENDPROOF

%% file: proof3.tex
\subsection{Proof of Theorem \ref{Thm_partial_2}}
In our proof, we distinguish three cases: $(\VeatureSize+\SecretSize)/2 > \NumOverlap \geq \VeatureSize -\SecretSize$, $\NumOverlap = \SecretSize$, and $\NumOverlap < \SecretSize$.

Assume that $\AlgA$ is an algorithm that takes as input two vault records $\VaultPoly$ and $\WaultPoly$ of feature sets $\VeatureSet$ and $\WeatureSet$ of size $\VeatureSize$ and $\WeatureSize$, respectively, with $|\VeatureSet \cap \WeatureSet|=\NumOverlap$, and outputs an element $\VeatureElement$ from the set difference $(\VeatureSet \cup \WeatureSet) \setminus(\VeatureSet \cap \WeatureSet)$ with probability $\prob{\AlgA}$. We will use $\AlgA$ to construct a full recover attack $\AlgB$ on ($\VaultPoly, \WaultPoly)$ that outputs ($\VeatureSet, \WeatureSet)$, and using Theorem \ref{EntropyLoss}, we will derive a bound on the success probability $\prob{\AlgA}$ of $\AlgA$.

Since the feature sets $\VeatureSet$ and $\WeatureSet$ are chosen uniformly at random, we have $\MinEntropy(\VeatureSet,\WeatureSet)=\tbinom{\FieldSize}{\VeatureSize+\WeatureSize-\NumOverlap}$. 

\BEGINPROOF[Proof for the case $(\VeatureSize+\SecretSize)/2 > \NumOverlap \geq \VeatureSize -\SecretSize$] 
Algorithm $\AlgB$ works as follows. Note, that since $(\VeatureSize+\SecretSize)/2 > \NumOverlap$, we have $ \ThreshDist \geq 1$.
\begin{enumerate}

\item Invoke algorithm $\AlgA$ with input ($\VeatureSet, \WeatureSet)$ to obtain an output $\VeatureElement_0\in (\VeatureSet \cup \WeatureSet) \setminus(\VeatureSet \cap \WeatureSet)$. 

\item Randomly choose a bit $x\in\{0,1\}$. If $x=0$, assume $\VeatureElement_0 \in \VeatureSet\setminus\WeatureSet$ and proceed with the subsequent steps as described below. If $x=1$, assume $\VeatureElement_0\in \WeatureSet\setminus\VeatureSet$ and proceed with the subsequent steps with $\VeatureSet$ and $\WeatureSet$ exchanged (and $\VeatureSize$ and $\WeatureSize$ exchanged likewise). 

\item Guess $\ThreshDist-1$ elements $\VeatureElement_1,\ldots, \VeatureElement_{\ThreshDist-1}$ from $\VeatureSet\setminus\WeatureSet$ and use (\ref{eq:reduction_coeff}) to iteratively compute the coefficients $(\bar{\VaultCoeff}_{\SecretSize-\ThreshDist}, \ldots, \bar{\VaultCoeff}_{\VeatureSize-\ThreshDist})$ of the characteristic polynomial $\CharPoly_{\bar{\VeatureSet}}$ of $\bar{\VeatureSet} = \VeatureSet \setminus \{\VeatureElement_0,\ldots,\VeatureElement_{\ThreshDist-1}\}$ from the coefficients $\VaultCoeff_{\SecretSize},\ldots,\VaultCoeff_{\VeatureSize}$ of $\VaultPoly$.

\item Guess $\ThreshDist$ elements $\WeatureElement_0,\ldots, \WeatureElement_{\ThreshDist-1}$ from $\WeatureSet\setminus\VeatureSet$ and use (\ref{eq:reduction_coeff}) to iteratively compute the coefficients $(\bar{\WaultCoeff}_{\SecretSize-\ThreshDist}, \ldots, \bar{\WaultCoeff}_{\WeatureSize-\ThreshDist})$ of the characteristic polynomial $\CharPoly_{\bar{\WeatureSet}}$ of $\bar{\WeatureSet} = \WeatureSet \setminus \{\WeatureElement_0,\ldots,\WeatureElement_{\ThreshDist-1}\}$ from the coefficients $\WaultCoeff_{\SecretSize},\ldots,\WaultCoeff_{\WeatureSize}$ of $\WaultPoly$ .
 
\item Invoke the partial recovery attack (Algorithm \ref{Algo:partial_recovery}) with parameters $\bar{\VeatureSize}=\VeatureSize-\ThreshDist$, $\bar{\WeatureSize}=\WeatureSize-\ThreshDist$ and $\bar{\SecretSize}=\SecretSize-\ThreshDist$
on input $(\bar{\VaultCoeff}_{\bar{\SecretSize}}, \ldots, \bar{\VaultCoeff}_{\bar{\VeatureSize}})$ and  $(\bar{\WaultCoeff}_{\bar{\SecretSize}}, \ldots, \bar{\WaultCoeff}_{\bar{\WeatureSize}})$. If the partial recovery attack returns \AlgFailure, do so as well and stop; otherwise assume that the output $(\NumOverlap^*,\bar{\VeatureSet}_0,\bar{\WeatureSet}_0)$ satisfies $\bar{\VeatureSet}_0=\bar{\VeatureSet}\setminus \bar{\WeatureSet}$ and $\bar{\WeatureSet}_0=\bar{\WeatureSet}\setminus \bar{\VeatureSet}$ and continue. 

\item Guess $\LastGuess=\SecretSize-\VeatureSize+\NumOverlap$ many elements $\VeatureElement_{\ThreshDist},\ldots,\VeatureElement_{\ThreshDist+m-1}$ from $\VeatureSet \cap \WeatureSet$. 

\item Unlock $\VeatureSet$: compute the unique polynomial $\VaultSecret^*$ of degree smaller than $\SecretSize$ interpolating the pairs $(\VeatureElement,\VaultPoly(\VeatureElement))$ for all $\VeatureElement\in \{\VeatureElement_0,\ldots, \VeatureElement_{\ThreshDist+m-1}\}\cup \bar{\VeatureSet}_0$. If $\VaultPoly-\VaultSecret^*$ splits into distinct linear factors, set $\VeatureSet$ as the set of its roots; otherwise return \AlgFailure. 

\item Set $\WeatureSet=\WeatureSet_0 \cup (\VeatureSet\setminus\VeatureSet_0)$. Output $(\VeatureSet,\WeatureSet)$.
\end{enumerate}

Assume that  $\AlgA$ has been successful and that the guess of $\AlgB$ on whether $\VeatureElement_0\in \VeatureSet\setminus\WeatureSet$ or $\VeatureElement_0\in \WeatureSet\setminus\VeatureSet$ was correct; this occurs with probability $\prob{\AlgA}/2$. 

If $\VeatureElement\in \VeatureSet\setminus\WeatureSet$, the probability $\prob{Guess1}$ that $\VeatureElement_1,\ldots, \VeatureElement_{\ThreshDist-1}$ and $\WeatureElement_0,\ldots, \WeatureElement_{\ThreshDist-1}$ are indeed elements of $\VeatureSet\setminus\WeatureSet$ and $\WeatureSet\setminus\VeatureSet$, respectively, is lower bounded by 
\begin{equation*}\label{eq:p_guess}
\prob{Guess1} \geq \binom{\VeatureSize-\NumOverlap}{\ThreshDist-1} \binom{\WeatureSize-\NumOverlap}{\ThreshDist} \binom{\FieldSize}{2\ThreshDist-1}^{-1}.
\end{equation*}
In the alternative case $\VeatureElement\in \WeatureSet\setminus\VeatureSet$, $\AlgB$ guesses $\ThreshDist$ elements from $\VeatureSet\setminus\WeatureSet$ and $\ThreshDist-1$ elements from $\WeatureSet\setminus\VeatureSet$, but since $\VeatureSize \geq \WeatureSize$, the lower bound (\ref{eq:p_guess}) for $\prob{Guess1}$ still holds. 

Assuming that all  $\VeatureElement_i$ and $\WeatureElement_j$ are elements of $\VeatureSet\setminus\WeatureSet$ and $\WeatureSet\setminus\VeatureSet$, respectively, input to the partial recovery attack are two vault records of the feature sets 
$\bar{\VeatureSet}$ and $\bar{\WeatureSet}$ of size $\bar{\VeatureSize}$ and $\bar{\WeatureSize}$, respectively, with the secret polynomials $\bar{\VaultSecret}$ and $\bar{\WaultSecret}$ having degree $\bar{\SecretSize}=\SecretSize-\ThreshDist$. Furthermore, if $\VeatureElement\in \VeatureSet\setminus\WeatureSet$, we have $\bar{\VeatureSize}\geq \bar{\WeatureSize}$ and 
$|\bar{\VeatureSet}  \cap \bar{\WeatureSet}| = \NumOverlap = (\bar{\VeatureSize}+\bar{\SecretSize})/2$. In the other case, \ie{} if $\VeatureElement\in \WeatureSet\setminus\VeatureSet$, we have $\bar{\WeatureSize}\geq \bar{\VeatureSize}$ and 
$|\bar{\VeatureSet}  \cap \bar{\WeatureSet}| = \NumOverlap = (\bar{\WeatureSize}+\bar{\SecretSize})/2$. Thus, in both cases, the output of the partial recovery attack satisfies $\bar{\VeatureSet}_0=\bar{\VeatureSet}\setminus \bar{\WeatureSet}$ and $\bar{\WeatureSet}_0=\bar{\WeatureSet}\setminus \bar{\VeatureSet}$. 

At this point, if all steps have been successful, algorithm $\AlgB$ has learned all $\VeatureSize-\NumOverlap$ elements of $\VeatureSet\setminus\WeatureSet$. Therefore, after guessing $\LastGuess$ many elements from  $\VeatureSet \cap \WeatureSet$, the unlocking procedure uncovers $\VeatureSet$, which also reveals $\WeatureSet$. (Since $\NumOverlap \geq \VeatureSize -\SecretSize$, we have $\LastGuess \geq 0$.) The guessing succeeds with probability $\prob{Guess2} \geq \tbinom{\NumOverlap}{\SecretSize-\VeatureSize+\NumOverlap} \tbinom{\FieldSize}{\SecretSize-\VeatureSize+\NumOverlap}^{-1}$.

Overall, the success probability $\prob{\AlgB}$ of algorithm $\AlgB$ is at least 
\begin{eqnarray}
\prob{\AlgB} & \geq &  (\prob{\AlgA}/2) \cdot \prob{Guess1} \cdot \prob{Guess2} \nonumber \\[0.6 em]
 & > & (\prob{\AlgA}/2) \frac{\binom{\VeatureSize-\NumOverlap}{\ThreshDist-1} \binom{\WeatureSize-\NumOverlap}{\ThreshDist} \binom{\NumOverlap}{\SecretSize-\VeatureSize+\NumOverlap} }
{\binom{\FieldSize}{2\ThreshDist-1} \binom{\FieldSize}{\SecretSize-\VeatureSize+\NumOverlap}}.\label{eq:thm3_lower}
\end{eqnarray}

On the other hand, Theorem \ref{EntropyLoss} gives
\begin{equation}\label{eq:thm3_upper}
\prob{\AlgB}  \leq \FieldSize^{\VeatureSize+\WeatureSize-2\SecretSize} \binom{\FieldSize}{\VeatureSize+\WeatureSize-\NumOverlap}^{-1}.
\end{equation}

Combining the inequalities (\ref{eq:thm3_lower}) and (\ref{eq:thm3_upper}) yields 
\begin{equation*}
\prob{}  \leq  
\frac{2 \FieldSize^{\VeatureSize+\WeatureSize-2\SecretSize} \binom{\FieldSize}{2\ThreshDist-1} \binom{\FieldSize}{\SecretSize-\VeatureSize+\NumOverlap}}
{\binom{\VeatureSize-\NumOverlap}{\ThreshDist-1} \binom{\WeatureSize-\NumOverlap}{\ThreshDist} \binom{\NumOverlap}{\SecretSize-\WeatureSize+\NumOverlap}\binom{\FieldSize}{\VeatureSize+\WeatureSize-\NumOverlap} } 
\end{equation*}
where $\ThreshDist=(\VeatureSize+\SecretSize)/2-\NumOverlap$. This completes the proof for the case $(\VeatureSize+\SecretSize)/2 > \NumOverlap \geq \VeatureSize -\SecretSize$
\ENDPROOF

We now turn to the second case. 
\BEGINPROOF[Proof for the case $\NumOverlap=\SecretSize$]
Algorithm $\AlgB$ works as follows. 
\begin{enumerate}

\item Invoke algorithm $\AlgA$ with input ($\VeatureSet, \WeatureSet)$ to obtain an output $\VeatureElement_0\in (\VeatureSet \cup \WeatureSet) \setminus(\VeatureSet \cap \WeatureSet)$. 

\item Randomly choose a bit $x\in\{0,1\}$. If $x=0$, assume $\VeatureElement_0 \in \VeatureSet\setminus\WeatureSet$ and proceed with the subsequent steps as described below. If $x=1$, assume $\VeatureElement_0\in \WeatureSet\setminus\VeatureSet$ and proceed with the subsequent steps with $\VeatureSet$ and $\WeatureSet$ exchanged (and $\VeatureSize$ and $\WeatureSize$ exchanged likewise). 

\item Guess $\SecretSize-1$ many further elements $\{\VeatureElement_1,\ldots\VeatureElement_{\SecretSize-1}\}$ from $\VeatureSet$.

\item Unlock $\VeatureSet$: compute the unique polynomial $\VaultSecret^*$ of degree smaller than $\SecretSize$ interpolating the pairs $(\VeatureElement,\VaultPoly(\VeatureElement))$ for all $\VeatureElement\in \{\VeatureElement_0,\ldots\VeatureElement_{\SecretSize-1}\}$. If $\VaultPoly-\VaultSecret^*$ splits into distinct linear factors, set $\VeatureSet$ as the set of its roots; otherwise return \AlgFailure. 

\item Guess $\VeatureSet \cap \WeatureSet$ by randomly choosing a $\SecretSize$-element subset of $\VeatureSet\setminus \{\VeatureElement_0\}$. 

\item Unlock $\WeatureSet$: compute the unique polynomial $\WaultSecret^*$ of degree smaller than $\SecretSize$ interpolating the pairs $(\VeatureElement,\WaultPoly(\VeatureElement))$ for all $\VeatureElement\in \VeatureSet \cap \WeatureSet$. If $\WaultPoly-\WaultSecret^*$ splits into distinct linear factors, set $\WeatureSet$ as the set of its roots; otherwise return \AlgFailure. 

\item Output $(\VeatureSet,\WeatureSet)$.

\end{enumerate}

Obviously, if $\AlgA$ is successful and all guesses are correct, $\AlgB$ succeeds as well. 

The probability that $\AlgA$ is successful and that $\AlgB$ guesses correctly whether $\VeatureElement_0\in \VeatureSet\setminus\WeatureSet$ or $\VeatureElement_0\in \WeatureSet\setminus\VeatureSet$ is $\prob{\AlgA}/2$. Furthermore, the $\SecretSize-1$ additional elements of $\VeatureSet$ are guessed with probability $\tbinom{\VeatureSize-1}{\SecretSize-1}/\tbinom{\FieldSize-1}{\SecretSize-1}$. Finally, guessing   $\VeatureSet \cap \WeatureSet$  as a $\SecretSize$-element subset of $\VeatureSet\setminus \{\VeatureElement_0\}$ succeeds with probability $\tbinom{\VeatureSize-1}{\SecretSize}^{-1}$. Overall, this results in the following lower bound for the success probability $\prob{\AlgB}$ of $\AlgB$.
\begin{equation*}
\prob{\AlgB} \geq (\prob{\AlgA}/2) \frac{\SecretSize}  {\binom{\FieldSize}{\SecretSize-1}(\VeatureSize-\SecretSize+1)}
\end{equation*}

On the other hand, Theorem \ref{EntropyLoss} gives an upper bound of
\begin{equation*}
\prob{\AlgB}  \leq \FieldSize^{\VeatureSize+\WeatureSize-2\SecretSize} \binom{\FieldSize}{\VeatureSize+\WeatureSize-\SecretSize}^{-1}.
\end{equation*}

Combining upper and lower bounds, we obtain
\begin{equation*}
\prob{\AlgA}  \leq \frac{2\FieldSize^{\VeatureSize+\WeatureSize-2\SecretSize} \binom{\FieldSize}{\SecretSize-1}(\VeatureSize-\SecretSize+1)}
{\binom{\FieldSize}{\VeatureSize+\WeatureSize-\SecretSize} \SecretSize }.
\end{equation*}
This completes the proof for the case $\NumOverlap=\SecretSize$.
\ENDPROOF

We now turn to the last case.
\BEGINPROOF[Proof for the case $\NumOverlap<\SecretSize$]
Algorithm $\AlgB$ works as follows. 
\begin{enumerate}

\item Invoke algorithm $\AlgA$ with input ($\VeatureSet, \WeatureSet)$ to obtain an output $\VeatureElement_0\in (\VeatureSet \cup \WeatureSet) \setminus(\VeatureSet \cap \WeatureSet)$. 

\item Randomly choose a bit $x\in\{0,1\}$. If $x=0$, assume $\VeatureElement_0 \in \VeatureSet\setminus\WeatureSet$ and proceed with the subsequent steps as described below. If $x=1$, assume $\VeatureElement_0\in \WeatureSet\setminus\VeatureSet$ and proceed with the subsequent steps with $\VeatureSet$ and $\WeatureSet$ exchanged (and $\VeatureSize$ and $\WeatureSize$ exchanged likewise). 

\item Guess all $\NumOverlap$ elements from $\VeatureSet \cap \WeatureSet$. 

\item Guess $\SecretSize-\NumOverlap-1$ many additional elements $\{\VeatureElement_1,\ldots\VeatureElement_{\SecretSize-\NumOverlap-1}\}$ from $\VeatureSet\setminus\WeatureSet$, and $\SecretSize-\NumOverlap$ many elements $\{\WeatureElement_1,\ldots\WeatureElement_{\SecretSize-\NumOverlap}\}$ from $\WeatureSet\setminus\VeatureSet$.

\item Unlock $\VeatureSet$: compute the unique polynomial $\VaultSecret^*$ of degree smaller than $\SecretSize$ interpolating the pairs $(\VeatureElement,\VaultPoly(\VeatureElement))$ for all $\VeatureElement\in (\VeatureSet \cap \WeatureSet) \cup \{\VeatureElement_0,\ldots\VeatureElement_{\SecretSize-\NumOverlap-1}\}$. If $\VaultPoly-\VaultSecret^*$ splits into distinct linear factors, set $\VeatureSet$ as the set of its roots; otherwise return \AlgFailure. 

\item Unlock $\WeatureSet$: compute the unique polynomial $\WaultSecret^*$ of degree smaller than $\SecretSize$ interpolating the pairs $(\WeatureElement,\WaultPoly(\WeatureElement))$ for all $\WeatureElement\in (\VeatureSet \cap \WeatureSet) \cup \{\WeatureElement_1,\ldots\WeatureElement_{\SecretSize-\NumOverlap}\}$. If $\WaultPoly-\WaultSecret^*$ splits into distinct linear factors, set $\WeatureSet$ as the set of its roots; otherwise return \AlgFailure. 

\item Output $(\VeatureSet,\WeatureSet)$.

\end{enumerate}

Obviously, if $\AlgA$ is successful and all guesses are correct, $\AlgB$ succeeds as well. 

The probability that $\AlgA$ is successful and that $\AlgB$ guesses correctly whether $\VeatureElement_0\in \VeatureSet\setminus\WeatureSet$ or $\VeatureElement_0\in \WeatureSet\setminus\VeatureSet$ is $\prob{\AlgA}/2$. Furthermore, $\VeatureSet \cap \WeatureSet$ is guessed with probability at least $\tbinom{\FieldSize}{\NumOverlap}^{-1}$. Finally, the $\SecretSize-\NumOverlap-1$ elements of $\VeatureSet$ and the $\SecretSize-\NumOverlap$ elements of $\WeatureSet$ are guessed with probability at least $\tbinom{\VeatureSize-\NumOverlap-1}{\SecretSize-\NumOverlap-1} / \tbinom{\FieldSize}{\SecretSize-\NumOverlap-1}$ and $\tbinom{\WeatureSize-\NumOverlap}{\SecretSize-\NumOverlap} / \tbinom{\FieldSize}{\SecretSize-\NumOverlap}$, respectively.

Overall, this results in the following lower bound for the success probability $\prob{\AlgB}$ of $\AlgB$.
\begin{equation*}
\prob{\AlgB} \geq (\prob{\AlgA}/2) \frac{\binom{\VeatureSize-\NumOverlap-1}{\SecretSize-\NumOverlap-1} \binom{\WeatureSize-\NumOverlap}{\SecretSize-\NumOverlap} }  {\binom{\FieldSize}{\NumOverlap} \binom{\FieldSize}{\SecretSize-\NumOverlap-1} \binom{\FieldSize}{\SecretSize-\NumOverlap} }
\end{equation*}

On the other hand, Theorem \ref{EntropyLoss} gives an upper bound of
\begin{equation*}
\prob{\AlgB}  \leq \FieldSize^{\VeatureSize+\WeatureSize-2\SecretSize} \binom{\FieldSize}{\VeatureSize+\WeatureSize-\NumOverlap}^{-1}.
\end{equation*}

Combining upper and lower bounds, we obtain
\begin{equation*}
\prob{\AlgA}  \leq \frac{2\FieldSize^{\VeatureSize+\WeatureSize-2\SecretSize}\binom{\FieldSize}{\NumOverlap} \binom{\FieldSize}{\SecretSize-\NumOverlap-1} \binom{\FieldSize}{\SecretSize-\NumOverlap}  }
{\binom{\FieldSize}{\VeatureSize+\WeatureSize-\NumOverlap}\binom{\VeatureSize-\NumOverlap-1}{\SecretSize-\NumOverlap-1} \binom{\WeatureSize-\NumOverlap}{\SecretSize-\NumOverlap}  }.
\end{equation*}
This completes the proof for the case $\NumOverlap<\SecretSize$.
\ENDPROOF